\documentclass[a4paper,UKenglish,cleveref, autoref, thm-restate]{lipics-v2021}

\usepackage{mathabx}
\usepackage{cmll}
\usepackage{tikz}
\usetikzlibrary{cd}
\nolinenumbers
\usepackage{mathtools}
\usepackage{prftree}
\usepackage{stmaryrd}
\SetSymbolFont{stmry}{bold}{U}{stmry}{m}{n}
\usepackage{xcolor}

\newcommand{\real}{\mathsf{Real}}
\newcommand{\fst}{\mathsf{fst}}
\newcommand{\snd}{\mathsf{snd}}
\newcommand{\lam}[3]{\lambda #1\,{:}\,#2.\,#3}
\newcommand{\pair}[2]{\langle #1,#2 \rangle}
\newcommand{\type}{\mathbf{Type}}
\newcommand{\LL}{\Lambda_{\real}}
\newcommand{\prt}[1]{#1^{\bullet}}
\newcommand{\cterm}[1]{\mathbf{T}_{#1}}

\newcommand{\dlog}{\delta^{\mathrm{log}}}
\newcommand{\deq}{\delta^{\mathrm{eq}}}
\newcommand{\psem}[1]{\llparenthesis #1 \rrparenthesis}
\newcommand{\sem}[1]{\llbracket #1 \rrbracket}
\newcommand{\fsem}[1]{\{\hspace{-2.25pt}| #1 |\hspace{-2.25pt}\}}
\newcommand{\mmm}{\mbox{quasi${}^{2}$-metric}}
\newcommand{\MMM}{\mbox{Quasi${}^{2}$-Metric}}
\newcommand{\qqm}{\mathbf{Qqm}}
\newif\ifshort
\shortfalse
\newif\iflong
\longtrue
\newif\ifcolor
\colorfalse

\newcommand{\qleq}{\sqsubseteq}
\newcommand{\qgeq}{\sqsupseteq}
\newcommand{\DLR}{differential logical relation}

\ifcolor
  \newcommand{\shortv}[1]{\ifshort\color{red}#1\color{black}\fi}
  \newcommand{\longv}[1]{\iflong\color{blue}#1\color{black}\fi}
\else
  \newcommand{\shortv}[1]{\ifshort#1\fi}
  \newcommand{\longv}[1]{\iflong#1\fi}
\fi

\title{On The Metric Nature \\ of (Differential) Logical Relations} 

\titlerunning{On The Metric Nature of (Differential) Logical Relations} 

\author{Ugo Dal Lago}{University of Bologna \and INRIA Sophia Antipolis }{ugo.dallago@unibo.it}{https://orcid.org/0000-0002-1825-0097}{}

\author{Naohiko Hoshino}{Sojo University, Japan}{nhoshino@cis.sojo-u.ac.jp}{[orcid]}{}

\author{Paolo Pistone}{Universit\'e Claude Bernard Lyon 1, France}{paolo.pistone@ens-lyon.fr}{[orcid]}{}

\authorrunning{Ugo Dal Lago, Naohiko Hoshino, Paolo Pistone} 

\Copyright{Ugo Dal Lago, Naohiko Hoshino, Paolo Pistone} 


\category{} 

\relatedversion{} 

\EventEditors{John Q. Open and Joan R. Access}
\EventNoEds{2}
\EventLongTitle{42nd Conference on Very Important Topics (CVIT 2016)}
\EventShortTitle{CVIT 2016}
\EventAcronym{CVIT}
\EventYear{2016}
\EventDate{December 24--27, 2016}
\EventLocation{Little Whinging, United Kingdom}
\EventLogo{}
\SeriesVolume{42}
\ArticleNo{23}

\hideLIPIcs

\begin{document}

\maketitle

\begin{abstract}
Differential logical relations are a method to measure distances between
  higher-order programs. They differ from standard methods based on program metrics in that differences between functional programs are themselves functions, relating errors in input with errors in output, this way providing a more fine grained, contextual, information.
The aim of this paper is
  to clarify the metric nature of differential
  logical relations.  While previous work has shown that these do not give rise, in general, to (quasi-)metric spaces nor 
  to partial metric spaces, we show that the distance functions arising from such relations, that we call quasi-quasi-metrics,   can be related to both quasi-metrics and partial metrics, the latter being also captured by suitable relational definitions.
 Moreover, we exploit such connections 
 to deduce some new compositional reasoning principles for program
  differences.
\end{abstract}

\section{Introduction}

Program equivalence is a crucial concept in
program semantics, and ensures that different
implementations of a program produce
\emph{exactly} the same results under the same
conditions, i.e., in any environment. This concept
is fundamental in program verification, code
optimization, and for enabling reliable
refactoring: by proving that two programs are
equivalent, developers and compiler designers can
confidently replace one with the other, knowing
that the behavior and outcomes will remain
consistent. In this respect, guaranteeing that the
underlying notion of program equality is a
congruence is of paramount importance.

In the research communities mentioned above, however, it is known that comparing programs through a notion of equivalence without providing the possibility of measuring the distance between non-equivalent programs makes it impossible to validate many interesting and useful program transformations \cite{Mittal2016}. All this has generated interest around the concepts of program metrics and more generally around the study of techniques through which to quantitatively compare non-equivalent programs, so as, e.g., to validate those program transformations which do not introduce too much of an error \cite{Reed2010, Plotk}. 

What corresponds, in a quantitative context, to the concept of congruence? Once differences are measured by some (pseudo-)metric, a natural answer to this question is to require that any language construct does not increase distances, that is, that they are \emph{non-expansive}. Along with this, the standard properties of (pseudo-)metrics, like the triangle inequality $ d(x,z) +d(z,y)\geq d(x,y)$, provide general principles that are very useful in metric reasoning, replacing standard qualitative principles (e.g., in this case, transitivity $\mathrm{eq}(x, z)\land\mathrm{eq}(z, y)\vdash \mathrm{eq}(x, y)$).

Still, as already observed in many occasions \cite{DGY19,DLG21}, the restriction to language constructs that are non-expansive with respect to some purely numerical metric turns out too severe in practice.
On the one hand, the literature focusing on higher-order languages has mostly restricted its attention to linear or graded languages \cite{Reed2010, Gaboardi2017}, due to well-known difficulties in constructing metric models for full ``simply-typed'' languages \cite{Honsell2022}.
On the other hand, even if one restricts to a linear language, the usual metrics defined over functional types are hardly useful in practice, as they assign distances to functions $f,g$ via a comparison of their values in the worst case: for instance, as shown in \cite{DGY19}, the two maps $\lambda x.x, \lambda x.\sin(x):\mathsf{Real}\to\mathsf{Real}$, although behaving very closely around $0$, are typically assigned the distance $\infty$, since their values grow arbitrarily far from each other in the worst case.

The \emph{differential logical relations} \cite{DGY19,DLG21,PistoneLICS,DG22} have been introduced as a solution to the aforementioned problems. In this setting, which natively works for unrestricted higher-order languages, the distance between two programs is not necessarily given as a single number: for instance, two programs of functional type are far apart according to a function itself, which measures how the {error} in the output depends on the \emph{error} in the input, but also on the \emph{value} of the input itself. This way the notion of distance becomes sufficiently expressive, at the same time guaranteeing the possibility of compositional reasoning. This paradigm also scales to languages with duplication, recursion \cite{DLG21} and works even in presence of effects \cite{DG22}.

In the literature on program metrics, it has become common to consider metrics valued on arbitrary \emph{quantales} \cite{Hofmann2014, Stubbe2014}. This means that, as for the differential logical relations, the distance between two points needs not be a non-negative real, but can 
belong to any suitable algebra of ``quantities''. This has led to the study of different classes of quantale-valued metrics, each characterized by a particular formulation of the triangular law. Among this, \emph{quasi-metrics} \cite{Goubault-Larrecq_2013} and \emph{partial metrics} \cite{matthews, Stubbe2018} have been explored for the study of domains, even for higher-order languages \cite{Geoffroy2020, maestracci2025}. While the first obey the usual triangular inequality, or transitivity, the second obey a \emph{stronger} transitivity condition, also taking into account the replacement of standard reflexivity $d(x,x)=0$ by a weaker \emph{quasi-reflexivity} condition $d(x,x)\leq d(x,y)$, implying that a point need not be a distance zero from itself.

A natural question is thus: do the distances between programs that are obtained via differential logical relations constitute some form of (quantale-valued) metric? In particular, what forms do transitivity and reflexivity do these relations support? The original paper \cite{DGY19} defined symmetric differential logical relations and gave a very weak form of triangle inequality. Subsequent works, relating to the more natural asymmetric case, have either ignored the metric question \cite{DLG21,DG22} or shown that the distances produced must violate \emph{both} the reflexivity of quasi-metric and the strong transitivity of partial metrics \cite{Geoffroy2020,PistoneLICS}.
%

This paper aims at providing a bridge between 
current methods for higher-order program differences and the well-established literature on quantale-valued metrics. More specifically, we show that the distances produced by differential logical relations, that we call \emph{quasi-quasi-metrics} (or \emph{quasi$^2$-metrics}), satisfy the \emph{quasi-}reflexivity of partial metrics and the standard transitivity of quasi-metrics. Such metrics thus sit somehow \emph{in between} quasi-metrics and partial metrics. We will establish precise connections between all those. We also exploit these results to deduce some new principles of compositional reasoning about program differences arising from the different forms of transitivity at play.
Finally, we introduce a deductive system, inspired from the quantitative equational theories of Mardare et al.~\cite{Plotk}, to deduce differences between programs.

\subparagraph{Contributions} Our contributions can be summarized as follows:
\begin{itemize}
\item We introduce a new class of quantale-valued metrics, called quasi$^2$-metrics. We show that each such metric gives rise to two \emph{observational quasi-metrics} over programs, and can be seen as a relaxation of partial \emph{quasi-}metrics \cite{KUNZI2006}. This is in Section 3;
\item we establish the equivalence of the cartesian closed structure of quasi$^2$-metrics with the standard definition of differential logical relations. We also show that observational quasi-metrics as well as partial quasi-metrics can be captured by suitable families of logical relations. We exploit all such definitions to deduce some new compositional reasoning principles for program differences. This spans through Sections 4-7;
\item finally, we introduce an equational theory for program differences via a syntactic presentation of differential logical relations and we formulate two conjectures about the comparison of the different notions of program distances introduced. This is in Sections 8 and 9.
\end{itemize}

\shortv{Due to lack of space, we omit many proofs, which
can be found in \cite{longversion}.}

%
%



\section{From Logical Relations to Differential Logical Relations}
\label{sec:from-logic-relat}

In this section we recall how differential logical relations can be seen as a quantitative generalization of standard logical relations, at the same time highlighting the metric counterparts of qualitative notions like equivalences and preorders. Moreover, we introduce
 \emph{quasi-quasi-metrics} as the metric counterpart of \emph{quasi-}reflexive and transitive relations.
 
\subparagraph{Logical Relations}

The theory of logical relations is well-known and has been exploited in various directions to establish \emph{qualitative}
 properties of type systems, like e.g.~termination \cite{Girard1989}, bisimulation \cite{Sangiorgi2007} or parametricity \cite{Plotkin1993, Hermida2014}. 
The idea is to start from some basic binary relation $\rho_o\subseteq o\times o$ over the terms of some ground type $o$. The relation $\rho_o$ can then be \emph{lifted} to a family of binary relations $\rho_{\mathsf{A}}\subseteq \mathsf{A}\times \mathsf{A}$, where $\mathsf{A}$ varies over all simple types constructed starting from $o$ (indeed, one may consider recursive \cite{Dreyer2009}, polymorphic \cite{Reynolds1983, Plotkin1993} or monadic \cite{GL2002} types as well, but we here limit our discussion to simple types).
The lifting is defined inductively by: 
\begin{align}
  (\mathsf t,\mathsf t') \in
  \rho_{\mathsf{A}\times \mathsf{B}}
  &\  \iff  \
    (\fst(\mathsf t), \fst(\mathsf t'))
    \in \rho_{\mathsf A} \ \text{and} \
    (\snd(\mathsf t), \snd(\mathsf t'))
    \in \rho_{\mathsf B},
    \tag{$\land$}\label{eq:land} \\
  (\mathsf t , \mathsf t')
  \in \rho_{\mathsf{A}\Rightarrow \mathsf{B}}
  &\  \iff  \
    (\forall \mathsf s,\mathsf s'\in\mathsf A ) \ 
    (\mathsf s, \mathsf s')\in \rho_{\mathsf A}
    \ \Rightarrow \
    ( \mathsf{t} \, \mathsf{s},
    \mathsf{t'} \, \mathsf{s'})
    \in \rho_{\mathsf B}
    \tag{$\Rightarrow $}\label{eq:to}.
\end{align}
Typically, one wishes to establish a so-called \emph{fundamental lemma}, stating that well-typed programs $x:\mathsf A\vdash\mathsf t:\mathsf B$ \emph{preserve relations}. This means that, for \emph{any} choice of a family of logical relations $\rho_{\mathsf A}$ defined as above, one can prove 
\begin{align*}
(\forall \mathsf s,\mathsf s'\in \mathsf A)\
 (\mathsf s, \mathsf s' )\in \rho_{\mathsf A}\ \Rightarrow \ 
  (\mathsf{t}[\mathsf s/\mathsf{x}] ,
  \mathsf{t}[\mathsf s'/\mathsf{x}])\in \rho_{\mathsf B}.
\tag{Fundamental Lemma}
\end{align*}
Notice that this is equivalent to the instance of reflexivity $(\lambda \mathsf{x}.\mathsf t,  \lambda \mathsf{x}.
\mathsf t)\in \rho_{\mathsf A\Rightarrow \mathsf B}$. 

Of particular interest are the \emph{equivalence} relations (that is, those which are reflexive, symmetric and transitive) and the \emph{preorders} (that is, the reflexive and transitive ones). We here focus on the latter, as we will not consider symmetry in this paper (see Remark \ref{rem:symmetry}).
A fundamental observation is that the logical relation lifting preserves preorders (and indeed, equivalences): if $\rho_{\mathsf A}$ and $\rho_{\mathsf B}$ are reflexive and transitive, then $\rho_{\mathsf A\times \mathsf B}$ and $\rho_{\mathsf A\Rightarrow \mathsf B}$ are reflexive and transitive as well, \emph{provided that the fundamental lemma holds}. The case of the function space is the most interesting one: as we observed above, 
the reflexivity condition $(\mathsf{t},\mathsf{t})\in \rho_{\mathsf A\Rightarrow \mathsf B}$ coincides with the fact that the function $\mathsf{t}$ is relation-preserving; transitivity, instead, can be proved by combining relation-preservation, the reflexivity of $\rho_{\mathsf A}$ and the transitivity of $\rho_{\mathsf B}$.

Any logical relation $\rho\subseteq \mathsf A\times \mathsf A$ induces an equivalence $\equiv_{\rho}$, called the \emph{observational equivalence}, where $\mathsf t  \equiv_{\rho} \mathsf u$ iff for all $\mathsf s\in \mathsf A$, $(\mathsf s,\mathsf t)\in \rho$ iff $(\mathsf s,\mathsf u)\in \rho$. Intuitively, two terms $\mathsf t,\mathsf u$ are equivalent if the relation $\rho$ cannot distinguish them. For example, if the definition of $\rho_o$ on basic types only depends on the values $\mathsf t\Rightarrow ^*\mathsf v$ produced by terms, one can usually deduce that terms are indistinguishable from their associated values, that is $\mathsf t \equiv_{\rho}  \mathsf v$.  
In the absence of symmetry, one obtains two \emph{observational preorders} $\sqsubseteq^l_{\rho},\sqsubseteq^r_{\rho}\ \subseteq \mathsf A\times \mathsf A$ defined by: 
\begin{align*}
\mathsf s\ \sqsubseteq^l_{\rho} \ \mathsf  t \ &\iff  \ 
(\forall \mathsf u\in \mathsf A) \ ( \mathsf t, \mathsf u)\in \rho \ \Rightarrow \ (\mathsf s, \mathsf u)\in \rho,
\\
\mathsf s\ \sqsubseteq^r_{\rho} \ \mathsf  t \ &\iff  \ 
(\forall \mathsf u\in \mathsf A) \ ( \mathsf u, \mathsf s)\in \rho \ \Rightarrow \ (\mathsf u, \mathsf t)\in \rho.
\end{align*}
These preorders satisfy the following useful and easily provable properties:
\begin{proposition}\label{prop:obs_preorder}
For any binary relation $\rho\subseteq \mathsf A\times \mathsf A$ and $c\in \{l,r\}$, 
\begin{description}
\item[(i.)] $\sqsubseteq^c_{\rho} \ \supseteq \ \rho$ iff $\rho$ is transitive;
\item[(ii.)] $\sqsubseteq^c_{\rho} \ \subseteq \ \rho$ iff $\rho$ is reflexive;
\item[(iii.)] $\sqsubseteq^c_{\rho} \ = \ \rho$ iff $\rho$ is a preorder;
\item[(iv)] The following hold:
  \begin{align}
    (\forall \mathsf s,\mathsf t,\mathsf u
    \in \mathsf A)
    \ \mathsf s \ \sqsubseteq^l_{\rho}
    \ \mathsf t \ \land \ 
    (\mathsf t, \mathsf u )
    \in \rho \ \Rightarrow \
    (\mathsf s, \mathsf u) \in \rho,
    \tag{left transitivity}\\
    (\forall \mathsf s,\mathsf t,\mathsf u\in \mathsf A )\ (\mathsf s, \mathsf t )\in \rho \ \land \ \mathsf t \ \sqsubseteq^r_{\rho}\ \mathsf u \ \Rightarrow \  (\mathsf s, \mathsf u)\in \rho.
    \tag{right transitivity}
  \end{align}
\end{description}

\end{proposition}

The reason why we delve into these basic properties of preorders is that we will soon explore their (less trivial!) quantitative counterparts, that arise naturally in the theory of differential logical relations. In particular, the left and right transitivity conditions will correspond to \emph{stronger} variants of the triangular inequality for metric spaces.

Beyond preorders, we are interested in the following weaker notion:
\begin{definition}[quasi-preorder]
A relation $\preceq\ \subseteq \mathsf A\times \mathsf A$ is called a \emph{quasi-preorder} if it is transitive and (left-)quasi-reflexive, that is, $\mathsf t  \preceq  \mathsf u \ \Rightarrow \ \mathsf t \preceq  \mathsf t$.
\end{definition}
Quasi-preorders are obtained by weakening the reflexivity condition of preorders: intuitively, only the points which are smaller than someone are smaller than themselves. 
One can easily develop a theory of logical relations for quasi-preorders. The sole delicate point is that, in order to let such relations lift to function spaces, one has to slightly modify the relation lifting as follows:
\begin{align}\label{eq:tostar}
  (\mathsf t , \mathsf t')
  \in \rho_{\mathsf{A}\Rightarrow \mathsf{B}}
  &\  \iff  \
    (\forall \mathsf s,\mathsf s'\in\mathsf A )\ 
    (\mathsf s, \mathsf s')
    \in \rho_{\mathsf A} \ \Rightarrow \
    ( \mathsf{t} \, \mathsf{s},
    \mathsf{t'} \, \mathsf{s'})
    \in  \rho_{\mathsf B}\ \land \
    (\mathsf{t} \, \mathsf{s},
    \mathsf{t} \, \mathsf{s}')
    \in \rho_{\mathsf B}
    \tag{$\Rightarrow ^*$}. 
\end{align}
Compared to \eqref{eq:to}, \eqref{eq:tostar}
includes a second clause
$( \mathsf{t} \, \mathsf{s}, \mathsf{t} \,
\mathsf{s}')\in \rho_{\mathsf B}$ relating the
action of $\mathsf t$ on both $\mathsf s$ and
$\mathsf s'$. With this definition, one can easily
check that if $\rho_{\mathsf A},\rho_{\mathsf{B}}$ are
quasi-preorders, and the fundamental lemma holds,
then $\rho_{\mathsf A\times \mathsf B}$ and
$\rho_{\mathsf A\Rightarrow \mathsf B}$ are
quasi-preorders as well.

\subparagraph{Differential Logical Relations}

We now have all elements to discuss what happens when extending logical relations to a quantitative setting. 
Rather than considering binary relations $\rho\subseteq \mathsf A\times \mathsf A$ expressing that a certain property holds for two terms $\mathsf s,\mathsf t$ or not, we will consider \emph{ternary} relations $\rho\subseteq \mathsf A\times \mathcal Q_{\mathsf A}\times \mathsf A$, where $(\mathsf s,\mathsf a,\mathsf t)\in \rho$ indicates that a certain relation holds of $\mathsf s,\mathsf t$ to \emph{a certain extent}, quantified via $\mathsf a\in \mathcal Q_{\mathsf A}$. Here $\mathcal Q_{\mathsf A}$ is a \emph{quantale}, an algebraic structure (recalled in the next section) that captures several properties of quantities as expressed by e.g.~non-negative real numbers.

In fact, just like for standard logical relations, a differential logical relation $\rho_o\subseteq o \times \mathcal Q_{o}\times o$ on a ground type can be \emph{lifted} to a family of binary relations $\rho_{\mathsf{A}}\subseteq \mathsf{A}\times \mathcal Q_{\mathsf A}\times \mathsf{A}$ over simple types. First, we define, by induction, the quantales $\mathcal Q_{\mathsf A\times \mathsf B}=\mathcal Q_{\mathsf A}\times \mathcal Q_{\mathsf B}$
and  $\mathcal Q_{\mathsf A\Rightarrow  \mathsf B}=\mathsf A \Rightarrow (\mathcal Q_{\mathsf A}\rightarrowtriangle\mathcal Q_{\mathsf B})$, where $ \mathcal Q_{\mathsf A}\rightarrowtriangle\mathcal Q_{\mathsf B}$ is the quantale of monotone functions. We then define the lifting of $\rho_o$ by:
\begin{align*}
  ((\mathsf t,\mathsf u),
  (a, b),
  (\mathsf t',\mathsf u'))
  \in \rho_{\mathsf{A}\times \mathsf{B}}
  &\iff
    (\mathsf t, a , \mathsf t')
    \in \rho_{ \mathsf A} \ \text{and} \
    (\mathsf u, b,  \mathsf u')
    \in \rho_{ \mathsf A}, \\
  (\mathsf t, f,\mathsf t')
  \in \rho_{\mathsf{A}\Rightarrow  \mathsf{B}}
  &\iff
    (\forall \mathsf s,
    \mathsf s' \in \mathsf A,
    \forall  a\in \mathcal Q_{\mathsf A})
    \text{ if }
    (\mathsf s, a, \mathsf s')\in \rho_{\mathsf A},
    \text{ then }
  \\
  &\mathrel{\phantom{\iff}}
    (\mathsf{t} \, \mathsf{s},
    f(\mathsf t)( a), \mathsf{t} \, \mathsf{s'})
    \in \rho_{\mathsf B}
    \text{ and }
    (\mathsf{t} \, \mathsf{s},
    f(\mathsf t)(a), \mathsf{t'} \, \mathsf{s'})
    \in \rho_{\mathsf B}. 
\end{align*}
Notice that the definition of $\rho_{\mathsf{A}\Rightarrow  \mathsf{B}}$ closely imitates the clause \eqref{eq:tostar} for quasi-preorders.
Also observe that the quantale $\mathcal Q_{\mathsf A\Rightarrow  \mathsf B}$ for the function type is itself a set of functions relating terms of type $\mathsf A$ and quantities in $\mathcal Q_{\mathsf A}$ with quantities in $\mathcal Q_{\mathsf B}$. As we show in Section 5, this definition gives rise to an interpretation of the simply typed $\lambda$-calculus where a fundamental lemma holds under the following form: for all terms $x:\mathsf A\vdash \mathsf t:\mathsf B$ and choice of a family of differential logical relations $\rho_{\mathsf A}$ as above, there exists a map $ {\mathsf t}^\bullet:\mathsf A\Rightarrow(\mathcal Q_{\mathsf A}\rightarrowtriangle  \mathcal Q_{\mathsf B})$ such that 
\begin{equation}\label{eq:log_rel_dlr}
  (\forall \mathsf s,\mathsf s'
  \in \mathsf A,
  \forall  a \in \mathcal Q_{\mathsf A}) \ 
  (\mathsf s, a, \mathsf s')
  \in \rho_{\mathsf A}
  \ \Rightarrow \ (\mathsf{t} \, \mathsf{s},
  {\mathsf t}^\bullet(\mathsf s)(a),
  \mathsf{t} \, \mathsf{s}')
  \in \rho_{\mathsf B}
  \tag{fundamental lemma}.
\end{equation}
The function ${\mathsf t}^\bullet$ behaves like some sort of \emph{derivative} of $\mathsf t$: it relates errors in input with errors in output. This connection is investigated in more detail in \cite{DLG21,PistoneLICS}.

So far, everything works just as in the standard, qualitative, case. However, the quantitative setting is well visible when we consider the corresponding notions of equivalences and preorders. Recall that an (integral) quantale is, in particular, an ordered monoid $(\mathcal Q,+, 0, \leq)$ of which $0$ is the minimum element.  
For a \DLR{} $\rho\subseteq \mathsf A\times \mathcal Q_{\mathsf A}\times \mathsf A$, reflexivity, symmetry and transitivity translate into the following conditions:
\begin{align}
  (\forall \mathsf t\in \mathsf A) \
  & (\mathsf t,0,\mathsf t)
    \in \rho, \tag{reflexivity} \\
  (\forall \mathsf t,
  \mathsf u \in \mathsf A,
  \forall  a \in \mathcal Q_{\mathsf A}) \
  & (\mathsf t, a,\mathsf u)
    \in \rho \ \Rightarrow \
    (\mathsf u, a,\mathsf t)
    \in \rho, \tag{symmetry} \\
  (\forall \mathsf s, \mathsf t, \mathsf u
  \in \mathsf A,
  \forall  a, b \in \mathcal Q_{\mathsf A}) \
  & (\mathsf s, a,\mathsf t)
    \in \rho \ \land \
    (\mathsf t, b,\mathsf u)
    \in \rho\ \Rightarrow \
    (\mathsf s, a+ b,\mathsf u)
    \in \rho. \tag{transitivity}
\end{align}
It is clear then that equivalence relations translate, in the quantitative setting, into some kind of metric space. Similarly,  
the quantitative counterpart of preorders are the so-called \emph{quasi-metric spaces} \cite{Goubault-Larrecq_2013}, essentially, metrics without a symmetry condition, indeed a very well-studied class of metrics. 
In particular, we will show that, similarly to preorders, any ternary relation $\rho\subseteq \mathsf A\times \mathcal Q_{\mathsf A}\times \mathsf A$ gives rise to left and right \emph{observational quasi-metrics} $q^l_{\rho},q^r_{\rho}: \mathsf A\times \mathsf A\to  \mathcal Q_{\mathsf A}$ satisfying properties analogous to those of Proposition \ref{prop:obs_preorder}.

\begin{remark}\label{rem:symmetry}
While in the original definition \cite{DGY19} differential logical relations were symmetric, symmetry was abandoned in all subsequent works. The first reason is that several interesting notions of program difference, like e.g.~those arising from \emph{incremental computing} \cite{DLG21, Giarrusso2014, AP2019}, are not symmetric. A second reason is that the cartesian closure is problematic in presence of both quasi-reflexivity and symmetry \cite{PistoneLICS}.  

\end{remark}

There is, however, an important point on which \DLR{}s differ from standard logical relations: while the former lift preorders  well to all simple types, their quantitative counterpart, the quasi-metrics, are \emph{not} preserved by the higher-order lifting of \DLR{}s. Indeed, we observed that an essential ingredient in the lifting of the reflexivity property is the fundamental lemma; yet, in the framework of \DLR{}s, the fundamental lemma produces, for any term $\mathsf t:\mathsf A\Rightarrow  \mathsf B$, the ``reflexivity'' condition $(\mathsf t,{\mathsf t}^\bullet,\mathsf t)\in \rho_{\mathsf A\Rightarrow  \mathsf B}$, which differs from standard reflexivity in that the distance is ${\mathsf t}^\bullet$ and \emph{not} the minimum element $0$. 
This means that the metric structure arising from \DLR{} cannot be that of standard (quasi-)metric spaces. Rather, it must be something close to the \emph{partial} metric spaces \cite{matthews, Stubbe2018}, that is, metric spaces in which the condition $d(x,x)=0$ is replaced by the quasi-reflexivity condition $d(x,x)\leq d(x,y)$. We will discuss the connections with partial metric spaces in the next sections.

By replacing reflexivity with quasi-reflexivity, we obtain the quantitative counterpart of quasi-preorders, that we call \emph{quasi$^2$-metrics} (being ``quasi'' both in the sense of quasi-metrics, i.e.~the rejection of symmetry, and of quasi-preorders, i.e.~the weakening of reflexivity).

\begin{definition}
For a set $X$ and a quantale $\mathcal Q$, a relation $\rho\subseteq X\times \mathcal Q\times X$ is called \emph{quasi-quasi-metric} (or more concisely \emph{quasi$^2$-metric}) if it is transitive and satisfies the condition 
\begin{align}
(\forall x,y\in X,\forall  a\in \mathcal Q) \
(x, a, y)\in \rho\ \Rightarrow \ (x, a, x)\in \rho.
\tag{quasi-reflexivity}
\end{align}
\end{definition}
As shown in Section 4,
the quasi$^2$-metrics capture the properties of distances which are preserved by \DLR{}s: indeed, the argument showing that the quasi-preorders lift to all simple types scales well to the quantitative setting, showing that a quasi$^2$-metric on the base types gives rise to quasi$^2$-metrics on all simple types.

The obvious question, however, is: what are these quasi$^{2}$-metrics? How are they related to the more standard quasi-metrics and partial metrics? This is what we are going to do in the following section.

\section{Quasi$^2$-Metric Spaces}
\label{sec:quasi2-metric-spaces}

In this section we use the language of quantale-valued relations to explore the connections between the quasi$^2$-metrics introduced in the previous section and the more well-established notions of quasi-metric and partial quasi-metric spaces.

\subparagraph{Quantale-Valued Relations}

Let us recall that a {quantale} $\mathcal Q$ is a
complete lattice $(\mathcal Q,\sqsubseteq)$
endowed with a continuous monoidal operation
$\otimes$, with unit $1$. A quantale $\mathcal Q$
is \emph{unital} when $1= \top$ and
\emph{commutative} when $\otimes$ is commutative.
Suppose $\mathcal Q$ is commutative. Given
$x,y\in \mathcal Q$, their \emph{residual} is
defined as
$x \multimap y:= \bigvee\{z \in \mathcal{Q} \mid z\otimes x\qleq
y\}$ where $\qleq$ is the partial order of
$\mathcal{Q}$. Notice that $z\qleq x \multimap y$
iff $z\otimes x\qleq y$, and that
$(x\multimap y)\otimes x\qleq y\qleq x\multimap
(y\otimes x)$. A commutative quantale $\mathcal Q$
is \emph{divisible} \cite{Stubbe2018} if for all
$x,y\in \mathcal Q$, $x\qleq y$ holds iff
$y\otimes (y\multimap x)=x$. Equivalently,
$\mathcal Q$ is divisible iff, whenever
$x\qleq y$, there exists $z$ such that
$x=y\otimes z$.
In the following we will use $\mathcal Q$ to refer
to a commutative, unital and divisible quantales.

\begin{example}
  The \emph{Lawvere quantale} is formed by the
  non-negative extended reals $[0,+\infty]$ with
  the \emph{reversed} order $x\qleq y:=x\geq y$,
  and with addition as monoidal operation. Notice
  that the ordering of quantales is
  \emph{reversed} with respect to usual metric
  intuitions: the ``$0$'' element is the $\top$,
  joins correspond to taking $\inf$s, etc.
\end{example}

Given a quantale $\mathcal Q$ and sets $X,Y$, a
\emph{$\mathcal Q$-relation over $X,Y$} is a map
$s\colon X\times Y\to \mathcal Q$, which
can be visualized as a matrix with values in
$\mathcal Q$. For $\mathcal{Q}$-relations
$s,t \colon X \times Y \to \mathcal{Q}$, we write
$s \qleq t$ when $s(x,y) \qleq t(x,y)$ for all
$x \in X$ and $y \in Y$. Given
$\mathcal Q$-relations
$s\colon X\times Y\to \mathcal Q$,
$t\colon Y\times Z\to \mathcal Q$ and
$u\colon X\times Z\to \mathcal Q$,
$w\colon Z\times Y\to \mathcal Q$, we define the
$\mathcal Q$-relations
$ s\otimes t \colon X\times Z \to \mathcal
Q$ and
$u\multimap s\colon Z\times Y \to \mathcal
Q$ and
$s\multimapinv w\colon X\times
Z \to \mathcal Q$ via the two operations:
\begin{align*}
  (s\otimes t)(x,z)
  &=\bigvee_{y\in Y}s(x,y)\otimes t(y,z), \\
  (u\multimap s)(z,y)
  &= \bigwedge_{x\in X}
    u(x,z)\multimap s(x,y),\qquad
    (s\multimapinv w)(x,z)
  = \bigwedge_{y\in Y}  w(z,y)\multimap s(x,y).
\end{align*}
The monoidal product $\otimes$ and the residuals
$\multimap,\multimapinv$ of $\mathcal Q$-relations
satisfy properties analogous to residuals in
$\mathcal Q$,
e.g.~$s\otimes (s\multimap t)\qleq t$,
$(t\multimapinv s)\otimes s\qleq t$. It is
well-known that $\mathcal Q$-relations form a
category $\mathcal Q\textbf{Rel}$ whose objects
are sets and such that
$\mathcal Q\textbf{Rel}(X,Y)$ are the
$\mathcal Q$-relations from $X$ to $Y$. The
operation $s\otimes t$ is the composition operator
of this category, while the identities are the
relations defined as $\mathbf 1_X(x,x)= 1=\top$
and $\mathbf 1_X(x,y\neq x)= \bot$.

Finally, for any relation
$s\in \mathcal Q\mathbf{Rel}(X,X)$, define the
relations
$\Delta_1 s,\Delta_2 s\in \mathcal
Q\mathbf{Rel}(X,X)$ by
$\Delta_1s=s\circ \Delta\circ \pi_1$ and
$\Delta_2s=s\circ \Delta\circ \pi_2$, that is,
$\Delta_1 s(x,y):=s(x,x)$,
$\Delta_2 s(x,y)=s(y,y)$.

\subparagraph{Quasi$^2$- and Quasi-Metric Spaces}
For a relation $s\in \mathcal Q\textbf{Rel}(X,X)$, reflexivity $s(x,x)=1$ and transitivity $s(x,z) \otimes s(z,y) \qleq s(x,y)$ can be written more concisely as 
$s \qgeq \mathbf{1}_{X}$ and $s\otimes s\qleq s$. A relation $s$ satisfying both such properties is called a \emph{quasi-metric over $X$} (with values in $\mathcal Q$). 
The following construction generalizes the observational preorders to $\mathcal Q$-relations:

\begin{proposition}\label{prop:left-right}
For all $s\in \mathcal Q\mathbf{Rel}(X,X)$, the relations $q^l_s:=s\multimapinv s, q^r_s:= s\multimap s \in \mathcal Q\mathbf{Rel}(X,X)$ are quasi-metrics and, for $c\in \{l,r\}$, the following hold:
\begin{description}
\item[(i.)] $q^c_s \qgeq s$ iff $s$ is transitive;
\item[(ii.)] $q^c_s \qleq s$ iff $s$ is reflexive;
\item[(iii.)] $q^c_s=s$ iff $s$ is a quasi-metric;
\item[(iv)] $q_s^l\otimes s\qleq s$ and 
 $s\otimes q^r_s\qleq s$, that is, the following hold:
\begin{align*}
(\forall x,y,z\in X)\ \ q^l_s(x,z)\otimes s(z,y) & \qleq s(x,y), \tag{left transitivity}\\
(\forall x,y,z\in X) \ \  s(x,z) \otimes q^r_s(z,y)& \qleq s(x,y). \tag{right transitivity}
\end{align*}
\end{description}
\end{proposition}
We call the quasi-metrics $q^l_s,q^r_s$ the \emph{left and right observational quasi-metric of $s$}.

Quasi$^{2}$-metrics correspond to $\mathcal Q$-relations $s\in \mathcal Q\mathbf{Rel}(X,X)$ satisfying 
transitivity $s\otimes s\qleq s$ and quasi-reflexivity $ s\qleq \Delta_1 s$ (i.e.~$s(x,y)\qleq s(x,x)$).
From transitivity, we deduce that, for a quasi$^2$-metric $s$, both $q_s^l,q_s^r\qgeq s$ hold, that is, the observational quasi-metrics yield \emph{tighter} distances than $s$. This implies that left and right transitivity read as \emph{stronger} forms of the triangular inequality.
 In particular, the following alternative characterization of quasi$^2$-metrics holds:

\begin{proposition}\label{prop:q2m}
For any quasi-reflexive $s\in \mathcal Q\mathbf{Rel}(X,X)$, $s$ is a quasi$^2$-metric iff
there exists a quasi-metric $q\qgeq s$ such that either $s\otimes q\qleq s$ or $q\otimes s\qleq s$ holds. 
\end{proposition}
\longv{
\begin{proof}
Suppose there exists a quasi-metric $q\qgeq s$ such that $s\otimes q\qleq s$ holds. 
Then $s\otimes s\qleq s\otimes q\qleq s$, so $s$ is transitive. A similar argument works if $q$ is 
such that $q\otimes s\qleq s$.
Conversely, if $s$ is a quasi$^2$-metric it is enough to let $q:=q_s^r$ and use Proposition \ref{prop:left-right} (iv). 
\end{proof}
}

%

\subparagraph{Partial Metric Spaces}

Let us now discuss the connection with partial metric spaces. We here consider the non-symmetric variant of the partial metric spaces from \cite{matthews}, called partial \emph{quasi-}metric spaces (PQM) \cite{KUNZI2006}.
As we anticipated, these are metrics $p$ for which the usual reflexivity condition $p(x,x)=1$ is replaced by the weaker quasi-reflexivity condition $p(x,x)\qgeq p(x,y)$. However, unlike the quasi$^2$-metrics just discussed, PQMs satisfy a \emph{stronger} transitivity condition. When $\mathcal Q=[0,+\infty]$ is the Lawvere quantale, this condition reads as 
\begin{equation}\label{eq:pqm}
p(x,z)+p(z,y)-p(z,z)\qgeq p(x,y).
\tag{strong transitivity in $[0,+\infty]$}
\end{equation}
The idea is that the self-distance of the central term $z$ is ``subtracted''. For a general quantale $\mathcal Q$, this becomes:
\begin{equation}\label{eq:pqm4}
p(x,z)\otimes(p(z,z)\multimap p(z,y)) \qleq p(x,y).
\tag{strong transitivity}
\end{equation}

 Define the relations $\Theta^l_s, \Theta^r_s\in \mathcal Q\mathbf{Rel}(X,X)$ by $\Theta^l_{s}(x,y)= s(y,y)\multimap s(x,y)$ and 
 $\Theta^r_s(x,y)=s(x,x)\multimap s(x,y)$. A PQM can be thus more concisely be defined as a relation $s\in \mathcal Q\mathbf{Rel}(X,X)$ satisfying $s\qleq \Delta_1 s$ and $s\otimes \Theta^r_s\qleq s$.
 Notice that strong transitivity $s\otimes \Theta^r_s\qleq s$ looks similar to the right transitivity $s\otimes q^r_s\qleq s$. 
Indeed, the following result relates the relations $\Theta^c_s$ and $q_s^c$:


\begin{proposition}\label{prop:quasipartial}
For all $s\in \mathcal Q\mathbf{Rel}(X,X)$ and $c\in\{l,r\}$, $q^c_s\qleq  \Theta^c_s$. Moreover, if $s$ is quasi-reflexive, 
$\Theta^c_s\qleq q^c_s$ holds iff 
$\Theta^c_s$ is a quasi-metric iff $s$ is a partial quasi-metric.

\end{proposition}
\longv{\begin{proof}
We only argue for $c=r$, the other case being similar.
From 
$q^r_s(x,y)= \bigwedge_z s(z,x)\multimap s(z,y)\qleq s(x,x)\multimap s(x,y)=(\Theta^r_s)(x,y)$ we deduce that 
$q^r_s\qleq \Theta^r_s$.
The converse direction $q^r_s\qgeq \Theta^r_s$ corresponds to showing that $s(z,x)\otimes(s(x,x)\multimap s(x,y))\qleq s(z,y)$, which holds iff $s$ is a partial quasi-metric. We have thus shown that $s$ is a PQM iff $\Theta^r_s=q^r_s$.
This also implies that, if $s$ is a PWM, $\Theta^r_s$ is a quasi-metric. Finally, suppose $\Theta^r_s$ is a quasi-metric.
By quasi-reflexivity, and the divisibility of $\mathcal Q$, we have that $s(x,z)=s(x,x)\otimes(s(x,x)\multimap s(x,z))$.
We then have $s(x,z)\otimes(s(z,z)\multimap s(z,y))= s(x,x)\otimes (s(x,x)\multimap s(x,z))\otimes (s(z,z)\multimap s(z,y)) \qleq  s(x,y)$, so $s$ is a partial quasi-metric. 
\end{proof}
}

The result above suggests that the partial quasi-metrics can be seen as limit cases of the quasi$^2$-metrics, namely those for which the quasi-metric $q^r_s(x,y)$ can be written under the simpler form $\Theta^r_s(x,y)=s(x,x)\multimap s(x,y)$. 

Unfortunately, while the standard definition of \DLR{}s preserves quasi$^2$-metrics, it does \emph{not} preserve partial quasi-metrics: \cite{Geoffroy2020, PistoneLICS} show that the function space constructions lifts PQMs into PQMs only when the monoidal product of the underlying quantales is \emph{idempotent} (one talks in this case of a partial \emph{ultra}-metric, since strong transitivity becomes $p(x,z)\land p(z,y)\qleq p(x,y)$).
Nevertheless, we will show in Section 6 how one can capture PQMs via a suitable family of logical relations.


\section{Differential Logical Relations as Quasi$^2$-Metrics}
\label{sec:category-weak-quasi}

In this section we provide a semantic presentation of differential logical relations by defining
a cartesian closed category of quasi$^2$-metrics, this way highlighting the close correspondence between these two notions.

\subparagraph{From $\mathcal Q$-Relations to Ternary Relations}
While in the previous section we discussed $\mathcal Q$-relations, that is, \emph{binary} relations valued in a quantale $\mathcal Q$, the theory of differential logical relations is expressed in terms of \emph{ternary} relations $\rho\subseteq X\times\mathcal Q\times X$.
In fact, any such relation $\rho\subseteq X\times \mathcal Q\times X$ induces a $\mathcal Q$-relation 
$\widehat \rho\in \mathcal Q\mathbf{Rel}(X,X)$ defined by 
\[
\widehat{\rho}(x,y)=\bigvee\{ a\in \mathcal{Q}\mid (x, a,y)\in \rho\}.
\]
Intuitively, $\widehat{\rho}(x,y)$ is the \emph{smallest} (recall the inversion of the order) distance between $x$ and $y$. 
This correspondence can be made more precise as follows: let a ternary relation $\rho\subseteq X\times \mathcal Q\times X$ be said \emph{$\mathcal Q$-closed} when the following hold:
\begin{itemize}
\item $(x,a,y)\in\rho$ and $a'\sqsubseteq a$ implies $(x,a',y)\in \rho$;
\item if $(x,a_i,y)\in \rho$, for all $i\in I$, then $(x,\bigvee_{i\in I}a_i,y)\in \rho$.
\end{itemize}

\longv{
We have the following correspondence:
\begin{lemma}\label{lemma:qclosure}
The map $\rho\mapsto \widehat\rho$ defines a bijection between the $\mathcal Q$-closed relations $\rho\subseteq X\times \mathcal Q\times X$ and $\mathcal Q\mathbf{Rel}(X,X)$.
\end{lemma}
\begin{proof}
Let $\rho,\tau$ be closed and let $\widehat\rho(x,y)=\widehat\tau(x,y)$. Observe that, for all $x,y\in X$, by $\mathcal Q$-closure we have $(x,\widehat\rho(x,y),y)\in \rho$. 
Suppose now that $(x,a,y)\in \tau$, then $a\sqsubseteq \widehat\tau(x,y)=\widehat\rho(x,y)$, and from $(x,\widehat\rho(x,y),y)\in \rho$ and $a\sqsubseteq \widehat\rho(x,y)$ we deduce $(x,a,y)\in \rho$. By a similar argument we can also prove that $(x,a,y)\in\rho$ implies $(x,a,y)\in \tau$, so in the end $\rho=\tau$. We conclude then that the map $\rho\mapsto \widehat\rho$ is injective.
For surjectivity, observe that any $s\in\mathcal Q\mathbf{Rel}(X,X)$ induces a relation $(x,a,y)\in\rho^s$ iff $a\sqsubseteq s(x,y)$, so that $s=\widehat{\rho^s}$.
\end{proof}
}
%
%

\shortv{
It is easily checked that the map $\rho\mapsto \widehat\rho$ defines a bijection between the $\mathcal Q$-closed relations $\rho\subseteq X\times \mathcal Q\times X$ and $\mathcal Q\mathbf{Rel}(X,X)$, and that the \DLR{} lifting preserves $\mathcal Q$-closure. In the sequel, we will identify metrics with their corresponding $\mathcal Q$-closed relations.}
\longv{
In the sequel, we will identify metrics with their corresponding $\mathcal Q$-closed relations. 
}

\subparagraph{A Cartesian Closed Category of Quasi$^2$-Metrics}
\label{sec:cart-clos-struct}

We 
now 
define a category of \mmm{}s. Let
us recall notations from
Section~\ref{sec:from-logic-relat}. For sets $A$
and $B$, we denote the set of functions from $A$
to $B$ by $A \Rightarrow B$; for quantales
$\mathcal{Q}$ and $\mathcal{R}$, we denote the set
of monotone functions from $\mathcal{Q}$ to
$\mathcal{R}$ by
$\mathcal{Q} \rightarrowtriangle \mathcal{R}$.
Below, 
we write $f \cdot x$ for the
application of 
$f \colon A \to B$ to
$x \in A$, and we suppose that $(-) \cdot (-)$ is
left-associative, i.e., $f \cdot x \cdot y$ is an
abbreviation of $(f \cdot x) \cdot y$.

The category $\qqm$ of quasi$^2$-metrics is defined as follows:
\begin{itemize}
\item objects are triples
  $X = (\mathcal{Q}_{X}, |X|,\rho_{X})$ consisting
  of a quantale $\mathcal{Q}_{X}$, a set $|X|$ and
  a \mmm{}
  $\rho_{X} \subseteq |X| \times \mathcal{Q}_{X}
  \times |X|$;
\item morphisms from $X$ to $Y$ are triples
  $(f,a,f')$ consisting of functions
  $f,f' \colon |X| \to |Y|$ and
  $a \colon |X| \to (\mathcal{Q}_{X}
  \rightarrowtriangle \mathcal{Q}_{Y})$ such that
  for all $(x,b,x') \in \rho_{X}$, we have
  $(f \cdot x,a \cdot x \cdot b, f' \cdot x') \in
  \rho_{Y}$ and
  $(f \cdot x,a \cdot x \cdot b, f \cdot x') \in
  \rho_{Y}$.
\end{itemize}
The identity morphism on an object $X$ is
$(\mathrm{id}_{X},i_{X},\mathrm{id}_{X})$
consisting of the identity function
$\mathrm{id}_{X}$ on $|X|$ and a function
$i_{X} \colon |X| \to (\mathcal{Q}_{X}
\rightarrowtriangle \mathcal{Q}_{X})$ given by
$i_{X} \cdot x \cdot a = a$. The composition of
$(f,a,f') \colon X \to Y$ and
$(g,b,g') \colon Y \to Z$ is
$(g \circ f,c, g' \circ f')$ where
$c \colon |X| \to (\mathcal{Q}_{X}
\rightarrowtriangle \mathcal{Q}_{Z})$ is given by
$c \cdot x = (b \cdot (f \cdot x)) \circ (a \cdot
x)$.

\begin{proposition}
  The category $\qqm$ is cartesian closed.
\end{proposition}
The cartesian closed structure corresponds to the
construction of differential logical relations in
Section~\ref{sec:from-logic-relat}. The terminal
object $\top$ is $(\{\ast\},\{\ast\},\rho_{\top})$
where $\rho_{\top} = \{(\ast,\ast,\ast)\}$, and
the product of $X$ and $Y$ is
$X \times Y = (\mathcal{Q}_{X} \times
\mathcal{Q}_{Y}, |X| \times |Y|, \rho_{X \times
  Y}) $, where $\rho_{X \times Y}$ is given by
\begin{equation*}
  ((x,y),(a,b),(x',y'))
  \in
  \rho_{X \times Y}
  \iff
  (x,a,x') \in \rho_{X}
  \text{ and }
  (y,b,y') \in \rho_{Y}.
\end{equation*}
The exponential $X \Rightarrow Y$ is
given by
$(|X| \Rightarrow (\mathcal{Q}_{X}
\rightarrowtriangle \mathcal{Q}_{Y}), |X|
\Rightarrow |Y|, \rho_{X \Rightarrow Y})$
where
\begin{equation*}
  (f,a,f')
  \in
  \rho_{X \Rightarrow Y}
  \iff
  \text{for all }
  (x,b,x') \in \rho_{X}
  \text{ and }
  g \in \{f,f'\}, \,
  (f \cdot x, a \cdot x \cdot b,
  g \cdot x') \in \rho_{Y}.
\end{equation*}
Here, the quantale structure of
$\mathcal{Q}_{X \times Y}$ and
$\mathcal{Q}_{X \Rightarrow Y}$ are given by the
pointwise manner. \longv{The first projection
  from $X \times Y$ to $Y$ is given by
  $(\mathrm{proj}_{X,Y},\varpi_{X,Y},\mathrm{proj}_{X,Y})$
  where
  $\mathrm{proj}_{X,Y} \colon |X| \times |Y| \to
  |X|$ is the first projection, and
  \begin{equation*}
    \varpi_{X,Y} \colon |X| \times |Y| \to
    (\mathcal{Q}_{X} \times \mathcal{Q}_{Y} \rightarrowtriangle
    \mathcal{Q}_{X})
  \end{equation*}
  is given by
  $\varpi_{X,Y} \cdot (x,y) \cdot (a,b) = a$.
  The second projection is given in the same
  manner. The tupling of
  $(f,a,f') \colon Z \to X$ and
  $(g,b,g') \colon Z \to Y$ is
  $(\langle f,g \rangle, \langle a,b \rangle,
  \langle f',g' \rangle)$ where
  $\langle f,g \rangle \colon |Z| \to |X| \times
  |Y|$ and
  $\langle a,b \rangle \colon |Z| \to
  (\mathcal{Q}_{Z} \rightarrowtriangle
  \mathcal{Q}_{X} \times \mathcal{Q}_{Y})$ are
  the tupling of $f,g$ and
  $a,b$ respectively:
  \begin{equation*}
    \langle f,g \rangle \cdot z
    = (f \cdot z, g \cdot z),
    \qquad
    \langle a,b \rangle \cdot z \cdot c
    =
    (a \cdot z \cdot c,
    a \cdot z \cdot c).
  \end{equation*}
  The currying of
  $(f,a,f') \colon Z \times X \to Y$ is
  $(f^{\wedge},a^{\wedge},f'^{\wedge})$ where
  $f^{\wedge} \colon |Z| \to (|X| \Rightarrow
  |Y|)$ and
  $f'^{\wedge} \colon |Z| \to (|X| \Rightarrow
  |Y|)$ are the currying of the following
  functions
  \begin{equation*}
    f \colon |Z| \times |X| \to |Y|,
    \qquad
    f' \colon |Z| \times |X| \to |Y|,
  \end{equation*}
  and
  $a^{\wedge} \colon |Z| \to (\mathcal{Q}_{Z}
  \rightarrowtriangle \mathcal{Q}_{X \Rightarrow
    Y})$ is the currying of
  \begin{equation*}
    a \colon |Z| \times |X| \to (\mathcal{Q}_{Z \times
      X} \rightarrowtriangle \mathcal{Q}_{Y}),
  \end{equation*}
  namely,
  $a^{\wedge} \cdot z \cdot a \cdot x \cdot b$
  is defined to be $a \cdot (z,x) \cdot (a,b)$.
  The evaluation morphism
  \begin{equation*}
    (\mathrm{eval}_{X,Y},
    \varepsilon_{X,Y},
    \mathrm{eval}_{X,Y})
    \colon
    (X \Rightarrow Y)
    \times X
    \to Y 
  \end{equation*}
  consists of the evaluation function
  $\mathrm{eval}_{X,Y} \colon (|X| \Rightarrow
  |Y|) \times |X| \to |Y|$ and
  $\varepsilon_{X,Y}$ is given by
  \begin{equation*}
    \varepsilon_{X,Y}
    \cdot (a,b) \cdot (f,x)
    =
    a \cdot x \cdot b.
  \end{equation*}
}

\begin{example}\label{eg:cube}
  We define an object $R \in \qqm$ to be
  $(\mathbb{R},[0,+\infty],\rho_{R})$, where
  \begin{equation*}
    (x,a,x') \in \rho_{R}
    \iff
    |x - x'| \sqsupseteq a.
  \end{equation*}
  Observe that the distance
  $\widehat{\rho}_R:\mathbb R\times \mathbb
  R\to[0,+\infty]$ is just the Euclidean distance
  $\widehat{\rho}_R(x,y)=|y-x|$. For functions
  $f,g \colon \mathbb{R} \to \mathbb{R}$, an
  element $a \in \mathcal{Q}_{R \Rightarrow R}$
  satisfies $(f,a,g) \in \rho_{R \Rightarrow R}$
  if and only if
  we have
  \begin{math}
    a \cdot x \cdot b \qleq
    \bigwedge_{|x - y| \sqsupseteq b}
    |f \cdot x - g \cdot y|, 
  \end{math}
  i.e., $a$ bounds gaps between outputs of $f$ and
  $g$. In particular, we have
  $(f,\top,f) \in \rho_{R \Rightarrow R}$ if and
  only if $f$ is a constant function. We note that
  the largest element
  $\top \in \mathcal{Q}_{R \Rightarrow R}$ is given by
  $\top \cdot x \cdot b = 0$.
\end{example}


\section{The Fundamental Lemma}
\label{sec:fundamental-lemma}

In this section we establish the fundamental lemma
of differential logical relations for a simply
typed lambda calculus $\LL$, by relying on the
cartesian closed category $\qqm$ of
quasi$^2$-metrics. We then apply this result to
measure differences between functions.

\subparagraph{Syntax and Set-theoretic Semantics}
\label{sec:synt-equat-theory}

Our language $\LL$ comprises a type of real
numbers and first order functions on $\mathbb{R}$.
Let $\mathsf{Var}$ be a countably infinite set of
variables. We define \emph{type}s and \emph{term}s
as follows:
\begin{align*}
  \text{(type)}
  &&
     \mathsf{A},\mathsf{B}
  &\coloneq
    \real
    \mid \mathsf{A} \times \mathsf{B}
    \mid \mathsf{A} \Rightarrow \mathsf{B} , \\
  (\text{term})
  &&
     \mathsf{t},\mathsf{s}
  &\coloneq
    \mathsf{x} \in \mathsf{Var}
    \mid \underline{r}
    \mid \phi
    (\mathsf{t}_{1},\ldots,\mathsf{t}_{n})
    \mid \mathsf{t} \, \mathsf{s}
    \mid \lam{\mathsf{x}}{\mathsf{A}}{\mathsf{t}}
    \mid \pair{\mathsf{t}}{\mathsf{s}}
    \mid \fst(\mathsf{t})
    \mid \snd(\mathsf{s})
    .
\end{align*}
Here, $r$ varies over $\mathbb{R}$, and $\phi$
varies over the set of multi-arity functions on
$\mathbb{R}$, namely, $\phi$ is a function from
$\mathbb{R}^{n}$ to $\mathbb{R}$ for some
$n \in \mathbb{N}$. We call $n$ the \emph{arity}
of $\phi$, and we denote the arity of $\phi$ by
$\mathrm{ar}(\phi)$. We adopt the standard typing
rules given in Figure~\ref{fig:context}. Below, we
denote the set of types by $\type$ and the set of
closed terms of type $\mathsf{A}$ by
$\cterm{\mathsf{A}}$.

\begin{figure}
  \fbox{
    \begin{minipage}{.96\textwidth}
      \footnotesize \centering
      \begin{equation*}
        \begin{array}{c}
          \prftree{
          \mathsf{x}:\mathsf{A} \in \mathsf{\Gamma}
          }{
          \mathsf{\Gamma} \vdash
          \mathsf{x}:\mathsf{A}
          }
          \hspace{15pt}
          \prftree{
          r \in \mathbb{R}
          }{
          \mathsf{\Gamma}
          \vdash
          \underline{r}:\real
          }
          \hspace{15pt}
          \prftree{
          \mathsf{\Gamma} \vdash 
          \mathsf{t}_{1}
          :\real
          }{\ldots}{
          \mathsf{\Gamma} \vdash 
          \mathsf{t}_{\mathrm{ar}(\phi)}
          :\real
          }{
          \mathsf{\Gamma} \vdash 
          \phi
          (\mathsf{t}_{1},
          \ldots,
          \mathsf{t}_{\mathrm{ar}(\phi)})
          : \real
          }
          \hspace{15pt}
          \prftree{
          \mathsf{\Gamma} \vdash 
          \mathsf{t}
          : \mathsf{A} \Rightarrow \mathsf{B}
          }{
          \mathsf{\Gamma} \vdash 
          \mathsf{s}
          : \mathsf{A}
          }{
          \mathsf{\Gamma} \vdash 
          \mathsf{t} \,
          \mathsf{s}
          : \mathsf{B}
          }
          \\[7pt]
          \prftree{
          \mathsf{\Gamma},\mathsf{x}:\mathsf{A} \vdash 
          \mathsf{t} : \mathsf{B}
          }{
          \mathsf{\Gamma} \vdash 
          \lam{\mathsf{x}}{\mathsf{A}}{\mathsf{t}}
          : \mathsf{A} \Rightarrow \mathsf{B}
          }
          \qquad
          \prftree{
          \mathsf{\Gamma} \vdash 
          \mathsf{t} : \mathsf{A}
          }{
          \mathsf{\Gamma} \vdash 
          \mathsf{s} : \mathsf{B}
          }{
          \mathsf{\Gamma} \vdash 
          \pair{
          \mathsf{t}
          }{
          \mathsf{s}
          }
          : \mathsf{A} \times \mathsf{B}
          }
          \qquad
          \prftree{
          \mathsf{\Gamma} \vdash 
          \mathsf{t}
          :\mathsf{A} \times \mathsf{B}
          }{
          \mathsf{\Gamma} \vdash 
          \fst(\mathsf{t})
          :\mathsf{A}
          }
          \qquad
          \prftree{
          \mathsf{\Gamma} \vdash 
          \mathsf{t}
          :\mathsf{A} \times \mathsf{B}
          }{
          \mathsf{\Gamma} \vdash 
          \snd(\mathsf{t})
          :\mathsf{B}
          }
        \end{array}
      \end{equation*}
    \end{minipage}
  }
  \caption{Typing Rules}
  \label{fig:context}
\end{figure}

We denote the standard set theoretic interpretation
of $\LL$ by $\psem{-}$. (See \cite{cft} for
example.) To be concrete, the interpretation
$\psem{\mathsf{A}}$ of a type $\mathsf{A}$ is a
set inductively defined by
\begin{equation*}
  \psem{\real} = \mathbb{R},
  \qquad
  \psem{\mathsf{A} \times \mathsf{B}} =
  \psem{\mathsf{A}} \times \psem{\mathsf{B}},
  \qquad
  \psem{\mathsf{A} \Rightarrow \mathsf{B}} =
  \psem{\mathsf{A}} \Rightarrow \psem{\mathsf{B}};
\end{equation*}
and we interpret a term
$\mathsf{x}_{1}:\mathsf{A}_{1},\ldots,
\mathsf{x}_{n}:\mathsf{A}_{n} \vdash \mathsf{t} :
\mathsf{B}$ as a function $\psem{\mathsf{t}}$ from
$\psem{\mathsf{A}_{1}} \times \cdots \times
\psem{\mathsf{A}_{n}}$ to $\psem{\mathsf{B}}$.
\shortv{For the definition of $\psem{\mathsf{t}}$,
  see \cite{longversion}.} \longv{The function
  $\psem{\mathsf{t}}$ is given by induction on the
  derivation of
  $\mathsf{x}_{1}:\mathsf{A}_{1},\ldots,
  \mathsf{x}_{n}:\mathsf{A}_{n} \vdash \mathsf{t}
  : \mathsf{B}$ as follows:
  \begin{align*}
    \psem{\mathsf{x}_{i}}
    \cdot (x_{1},\ldots,x_{n})
    &= x_{i}, \\
    \psem{\underline{r}}
    \cdot (x_{1},\ldots,x_{n})
    &= r, \\
    \psem{\phi
    (\mathsf{t}_{1},\ldots,
    \mathsf{t}_{\mathrm{ar}(\phi)})}
    \cdot (x_{1},\ldots,x_{n})
    &= \phi(
      \psem{\mathsf{t}_{1}}
      \cdot (x_{1},\ldots,x_{n}),
      \ldots,
      \psem{\mathsf{t}_{\mathrm{ar}(\phi)}}
      \cdot (x_{1},\ldots,x_{n})
      ), \\
    \psem{\mathsf{t} \, \mathsf{s}}
    \cdot (x_{1},\ldots,x_{n})
    &=
      \psem{\mathsf{t}}
      \cdot
      (x_{1},\ldots,x_{n}) \cdot
      (\psem{\mathsf{s}}
      \cdot
      (x_{1},\ldots,x_{n})),
    \\
    \psem{\lam{\mathsf{x}}
    {\mathsf{A}}{\mathsf{t}}}
    \cdot (x_{1},\ldots,x_{n})
    \cdot y
    &= \psem{\mathsf{t}}
      \cdot (x_{1},\ldots,x_{n},y),
    \\
    \psem{\fst(\mathsf{t})}
    \cdot (x_{1},\ldots,x_{n})
    &=
      \text{the first component of }
      \psem{\mathsf{t}}
      \cdot (x_{1},\ldots,x_{n}), \\
    \psem{\snd(\mathsf{t})}
    \cdot (x_{1},\ldots,x_{n})
    &=
      \text{the second component of }
      \psem{\mathsf{t}}
      \cdot (x_{1},\ldots,x_{n}), \\
    \psem{\pair{\mathsf{t}}{\mathsf{s}}}
    \cdot (x_{1},\ldots,x_{n})
    &= (
      \psem{\mathsf{t}}
      \cdot (x_{1},\ldots,x_{n}),
      \psem{\mathsf{s}}
      \cdot (x_{1},\ldots,x_{n})
      ).
  \end{align*}}

\subparagraph{The Fundamental Lemma}

We inductively define a \mmm{} space
$\sem{\mathsf{A}}$ by
\begin{equation*}
  \sem{\real} = R,
  \qquad
  \sem{\mathsf{A} \times \mathsf{B}} =
  \sem{\mathsf{A}} \times \sem{\mathsf{B}},
  \qquad
  \sem{\mathsf{A} \Rightarrow \mathsf{B}} =
  \sem{\mathsf{A}} \Rightarrow \sem{\mathsf{B}},
\end{equation*}
and we simply denote the structure of an object
$\sem{\mathsf{A}}$ by
$(|\mathsf{A}|,\mathcal{Q}_{\mathsf{A}},\rho_{\mathsf{A}})$.
It is straightforward to check that for every type
$\mathsf{A}$, we have
$|\mathsf{A}| = \psem{\mathsf{A}}$.
The \mmm{}s $\sem{\mathsf{A}}$ are the categorical
interpretation of types $\mathsf{A}$, and the
following fundamental lemma is derived from the
categorical interpretation of $\LL$-terms in
$\qqm$.

\begin{theorem}[Fundamental Lemma]\label{thm:fl}
  Let
  $\mathsf{\Gamma} =
  (\mathsf{x}_{1}:\mathsf{A}_{1},\ldots,
  \mathsf{x}_{n}:\mathsf{A}_{n})$ be a typing
  context. For every term
  $\mathsf{\Gamma} \vdash \mathsf{t} :
  \mathsf{A}$, and for every
  $(x,a,x') \in \rho_{\mathsf{A}_{1} \times \cdots
    \times \mathsf{A}_{n}}$, we have
  \begin{equation*}
    (\psem{\mathsf{t}} \cdot x,
    \fsem{\mathsf{t}} \cdot x \cdot a,
    \psem{\mathsf{t}} \cdot x')
    \in \rho_{{\mathsf{A}}} 
  \end{equation*}
  where we inductively define
  $\fsem{\mathsf{t}} \in \mathcal
  Q_{\mathsf{A}_{1} \times \cdots \times
    \mathsf{A}_{n} \Rightarrow \mathsf{B}}$ as
  follows:
  \begin{itemize}
  \item We define
    $\fsem{\mathsf{x}_{i}} \cdot
    (x_{1},\ldots,x_{n}) \cdot
    (a_{1},\ldots,a_{n})$ to be
    \begin{math}
      a_{i}.
    \end{math}
  \item We define
    $\fsem{\underline{r}} \cdot x \cdot a$ to be
    \begin{math}
      \underline{0}.
    \end{math}
  \item We define
    $\fsem{\phi(\mathsf{t}_{1},\ldots,\mathsf{t}_{n})}
    \cdot x \cdot a$ to be
    \begin{math}
      \phi^{\mathrm{d}}( \psem{\mathsf{t}_{1}} \cdot x, \ldots,
      \psem{\mathsf{t}_{n}} \cdot x,
      \fsem{\mathsf{t}_{1}} \cdot x \cdot a,
      \ldots, \fsem{\mathsf{t}_{n}} \cdot x \cdot
      a)
    \end{math}
    where we define
    $\phi^{\mathrm{d}} \colon \mathbb{R}^{n}
    \times [0,+\infty]^{n} \to [0,+\infty]$ by
    \begin{equation*}
      \phi^{\mathrm{d}}(y_{1},\ldots,y_{n},b_{1},\ldots,b_{n})
      =
      \bigwedge_{|y_{1} - z_{1}| \sqsupseteq b_{1}}
      \cdots
      \bigwedge_{|y_{n} - z_{n}| \sqsupseteq b_{n}}
      |\phi(y_{1},\ldots,y_{n})
      - \phi(z_{1},\ldots,z_{n})|.
    \end{equation*}
  \item We define
    $\fsem{\mathsf{t} \, \mathsf{s}} \cdot x \cdot
    a$ to be
    $\fsem{\mathsf{t}} \cdot x \cdot a \cdot
    (\psem{\mathsf{s}} \cdot x) \cdot
    (\fsem{\mathsf{s}} \cdot x \cdot a)$.
  \item We define
    $\fsem{\lam{\mathsf{x}}{\mathsf{A}}{\mathsf{t}}}
    \cdot (x_{1},\ldots,x_{n}) \cdot
    (a_{1},\ldots,a_{n}) \cdot y \cdot b$ to be
    $\fsem{\mathsf{t}} \cdot
    (x_{1},\ldots,x_{n},y) \cdot
    (a_{1},\ldots,a_{n},b)$.
  \item We define
    $\fsem{\pair{\mathsf{t}}{\mathsf{s}}} \cdot x
    \cdot a$ to be
    $(\fsem{\mathsf{t}} \cdot x \cdot a,
    \fsem{\mathsf{s}} \cdot x \cdot a)$.
  \item We define
    $\fsem{\fst(\mathsf{t})} \cdot x \cdot a$ to
    be the first component of
    $\fsem{\mathsf{t}} \cdot x \cdot a$.
  \item We define
    $\fsem{\snd(\mathsf{t})} \cdot x \cdot a$ to
    be the second component of
    $\fsem{\mathsf{t}} \cdot x \cdot a$.
  \end{itemize}
\end{theorem}
\longv{
  \begin{proof}
    The triple
    $(\psem{\mathsf{t}}, \fsem{\mathsf{t}},
    \psem{\mathsf{t}})$ is the interpretation of
    $\mathsf{\Gamma} \vdash \mathsf{t} :
    \mathsf{A}$ in the cartesian closed category
    $\qqm$ where we interpret
    $\mathsf{\Gamma} \vdash
    \phi(\mathsf{t}_{1},\ldots,\mathsf{t}_{n})$ by
    \begin{equation*}
      \sem{\mathsf{A}_{1}}
      \times \cdots \times
      \sem{\mathsf{A}_{n}}
      \xrightarrow{
        \langle
        \sem{\mathsf{t}_{1}},
        \ldots,
        \sem{\mathsf{t}_{n}}
        \rangle
      }
      R \times \cdots \times R
      \xrightarrow{
        (\phi,\phi^{\mathrm{d}},\phi)
      }
      R.
    \end{equation*}
    The statement follows from that
    $(\psem{\mathsf{t}}, \fsem{\mathsf{t}},
    \psem{\mathsf{t}})$ is a morphism from
    $\sem{\mathsf{A}_{1}} \times \cdots \times
    \sem{\mathsf{A}_{n}}$ to $\sem{\mathsf{A}}$ in
    $\qqm$.
  \end{proof}
} 

The fundamental lemma is a way to compositionally
reason about distances.
\begin{example}\label{eg:phipsi}
  Let us fix a positive real number $\epsilon$. We
  define
  $D_\epsilon \colon (\mathbb{R} \Rightarrow
  \mathbb{R}) \to (\mathbb{R} \Rightarrow
  \mathbb{R})$ by
  \begin{equation*}
    D_\epsilon \cdot f =
    \lambda x:\mathbb{R}. \,
    \frac{f(x + \epsilon) - f(x)}{\epsilon}.
  \end{equation*}
  For $f \colon \mathbb{R} \to \mathbb{R}$ and
  $x \in \mathbb{R}$, $D_\epsilon \cdot f \cdot x$
  calculates an approximation of the derivative of
  $f$ at $x$. By the fundamental lemma, we obtain
  $(D_\epsilon,E_\epsilon,D_\epsilon) \in \rho_{(R
    \Rightarrow R) \Rightarrow (R \Rightarrow R)}$
  where $E_\epsilon$ is a function from
  $|R \Rightarrow R|$ to
  $\mathcal Q_{R \Rightarrow R}
  \rightarrowtriangle \mathcal Q_{R \Rightarrow
    R}$ given by
  \begin{equation*}
    E_\epsilon \cdot f \cdot a
    =
    \lambda x:\mathbb{R}. \,
    \lambda b:[0,+\infty]. \,
    \frac{a \cdot (x + \epsilon) \cdot b
      + a \cdot x \cdot b}{\epsilon}.
  \end{equation*}
  In Example~\ref{eg:id-sin}, we will observe that
  $(\mathrm{id}_{\mathbb{R}},a,\sin)$ is an
  element of $\rho_{R \Rightarrow R}$ where
  $\mathrm{id}_{\mathbb{R}}$ is the identity
  function on $\mathbb{R}$, and
  $a \in \mathcal{Q}_{R \Rightarrow R}$ is given
  by $a \cdot x \cdot b = |x - \sin(x)| + b$. By
  applying
  $(D_\epsilon,E_\epsilon,D_\epsilon)$
  to $(\mathrm{id}_{\mathbb{R}},a,\sin)$,
  we obtain
  $(D_{\epsilon} \cdot \mathrm{id}_{\mathbb{R}},
  a',
  D_{\epsilon} \cdot \sin) \in
  \rho_{R}$ where
  \begin{equation*}
    a' \cdot x \cdot b =
    \frac{|x + \epsilon - \sin(x + \epsilon)|
    +
    |x - \sin(x)|
    + 2b}{\epsilon}.
  \end{equation*}
  From this, we see that the distance between
  $D_{\epsilon} \cdot \mathrm{id}_{\mathbb{R}}
  \cdot 0$ and $D_{\epsilon} \cdot {\sin} \cdot 0$
  is bounded by
  $\frac{|\epsilon - \sin(\epsilon)|}{\epsilon}$.
  We note that $a'$ is not the exact distance
  between
  $D_{\epsilon} \cdot \mathrm{id}_{\mathbb{R}}$
  and $D_{\epsilon} \cdot \sin$. For example,
  while
  $|D_{0.1} \cdot \mathrm{id}_{\mathbb{R}} \cdot 0
  - D_{0.1} \cdot {\sin} \cdot 0.1| \approx 0.01$,
  we have $a' \cdot 0 \cdot 0.1 \approx 2$. This
  gap stems in the fact that
  $(D_{\epsilon}, E_{\epsilon},D_{\epsilon})$
  takes all functions into account and can not
  exploit continuity of specific functions.
\end{example}

\section{Quasi-Metric Logical Relations}
\label{sec:quasi-metr-assoc}

As described in Section 3, any quasi$^2$-metric
gives rise to left and right observational
quasi-metrics. In this section, we introduce a
class of logical relations $\gamma_{\mathsf A}$
that capture the left observational quasi-metric
associated to $\rho_{{\mathsf A}}$. We will then show
how such relations can be used to derive
over-approximations of distances between
functions.

For a type $\mathsf{A}$, we define
$\gamma_{\mathsf{A}} \subseteq |\mathsf{A}| \times
\mathcal{Q}_{\mathsf{A}} \times |\mathsf{A}|$ by
induction on $\mathsf{A}$ as follows:
\begin{align*}
  (x,a,x') \in
  \gamma_{\real}
  &\iff |x - x'| \sqsupseteq a, \\
  (f,a,f') \in
  \gamma_{\mathsf{A} \Rightarrow \mathsf{B}}
  &\iff \text{for all }
    (x,b,x') \in
    \rho_{\mathsf{A}}, \,
    (f \cdot x, a \cdot x \cdot b,
    f' \cdot x) \in \gamma_{\mathsf{B}},
    \text{ and}
  \\
  &\mathrel{\phantom{\iff}}
    \text{for all }
    (f',b,f')
    \in \rho_{\mathsf{A} \Rightarrow \mathsf{B}}, \,
    (f,a \otimes b, f) \in
    \rho_{\mathsf{A} \Rightarrow \mathsf{B}}, \\
  ((x,y),(a,b),(x',y'))
  \in
  \gamma_{\mathsf{A} \times \mathsf{B}}
  &\iff
    (x,a,x') \in \gamma_{\mathsf{A}}
    \text{ and }
    (y,b,y') \in \gamma_{\mathsf{B}}.
\end{align*}
We give some explanation on the definition of
$\gamma_{\mathsf{A} \Rightarrow \mathsf{B}}$. The
definition consists of two conditions. The first
condition means that, if $(f,a,f')$ is an element
of $\gamma_{\mathsf{A} \Rightarrow \mathsf{B}}$,
then the distance $a \cdot x \cdot b$
over-approximates the distance between $f$ and
$f'$ at \emph{the same} point $x$ (rather than on
distinct points, as is the case for the relation $\rho_{\mathsf A \Rightarrow \mathsf B}$). The
second condition means that $a$ also
over-approximates the \emph{gap} between the
self-distance of $f'$ and the self-distance of
$f$.
%

\longv{Let us introduce a notation. For a type
  $\mathsf{A}$ and $x \in |\mathsf{A}|$, we write
  $\lceil x \rceil \in \mathcal{Q}_{\mathsf{A}}$ for
  $\hat{\rho_{\mathsf{A}}}(x,x)$, i.e.,
  \begin{equation*}
    \lceil x \rceil
    = \sup
    \{a \in \mathcal{Q}_{\mathsf{A}} \mid
    (x,a,x) \in \rho_{\mathsf{A}}\}.
  \end{equation*}
  Since $\rho_{\mathsf{A}}$ is closed under
  supremum, for any $x \in |\mathsf{A}|$, we have
  $(x,\lceil x \rceil,x) \in \rho_{\mathsf{A}}$.
  \begin{lemma}\label{lem:ceil-refl}
    For every type $\mathsf{A}$ and for every
    $x,x' \in |\mathsf{A}|$, if
    $(x,\lceil x \rceil,x') \in
    \rho_{\mathsf{A}}$, then $x = x'$.
  \end{lemma}
  \begin{proof}
    By induction on $\mathsf{A}$. It is
    straightforward to check the case $\real$ and
    the case $\mathsf{A} \times \mathsf{B}$. For
    the case $\mathsf{A} \Rightarrow \mathsf{B}$,
    if
    $(f,\lceil f \rceil,f') \in \rho_{\mathsf{A}
      \Rightarrow \mathsf{B}}$, then for any
    $x \in |\mathsf{A}|$, we have
    \begin{equation*}
      (f \cdot x,
      \lceil f \rceil \cdot x \cdot \lceil x \rceil,
      f' \cdot x)
      \in \rho_{\mathsf{A}}.
    \end{equation*}
    Here, by the induction hypothesis,
    \begin{align*}
      \lceil f \rceil
      \cdot x \cdot \lceil x \rceil
      &=
        \sup \{a \in \mathcal{Q}_{\mathsf{B}} \mid
        \text{for all }
        (x,\lceil x \rceil, x')
        \in \rho_{\mathsf{A}}, \,
        (f \cdot x, a, f \cdot x') \in \rho_{\mathsf{B}}
        \} \\
      &=
        \sup \{a \in \mathcal{Q}_{\mathsf{B}} \mid
        (f \cdot x, a, f \cdot x) \in \rho_{\mathsf{B}}
        \} \\
      &= \lceil f \cdot x \rceil.
    \end{align*}
    Hence, $f' \cdot x = f \cdot x$.
  \end{proof}

  \begin{lemma}\label{lem:ceil-ceil}
    For any type $\mathsf{A}$ and $\mathsf{B}$, if
    $f \in |\mathsf{A} \Rightarrow \mathsf{B}|$
    and $x \in |\mathsf{A}|$, then
    $\lceil f \rceil \cdot
    x \cdot \lceil x \rceil= \lceil
    f \cdot x\rceil$.
  \end{lemma}
  \begin{proof}
    This is shown in the proof of
    Lemma~\ref{lem:ceil-refl}.
  \end{proof}
  \begin{lemma}\label{lem:rho-gamma}
    For every type $\mathsf{A}$, if
    $(x,a \otimes \lceil x' \rceil,x') \in
    \rho_{\mathsf{A}}$, then
    $(x,a,x') \in \gamma_{\mathsf{A}}$.
  \end{lemma}
  \begin{proof}
    By induction on $\mathsf{A}$. The only
    non-trivial case is
    $\mathsf{A} \Rightarrow \mathsf{B}$. If
    $(f,a \otimes \lceil f' \rceil,f') \in
    \rho_{\mathsf{A} \Rightarrow \mathsf{B}}$,
    then for any $(x,b,x') \in \rho_{\mathsf{A}}$,
    since
    $(x, b \vee \lceil x \rceil,x) \in
    \rho_{\mathsf{A}}$, we obtain
    \begin{equation*}
      (f \cdot x,
      (a \cdot x \cdot
      (b \vee \lceil x \rceil))
      \otimes
      (\lceil f' \rceil
      \cdot x \cdot (b \vee \lceil x \rceil)),
      f' \cdot x)
      \in
      \rho_{\mathsf{B}}.
    \end{equation*}
    It follows from monotonicity of $a$
    and Lemma~\ref{lem:ceil-ceil}
    that we have
    \begin{equation*}
      (f \cdot x,
      (a \cdot x \cdot b)
      \otimes
      (\lceil f' \rceil
      \cdot x \cdot
      \lceil x \rceil),
      f' \cdot x)
      =
      (f \cdot x,
      (a \cdot x \cdot b) \otimes
      \lceil f' \cdot x \rceil,
      f' \cdot x)
      \in \rho_{\mathsf{B}}.
    \end{equation*}
    By the induction hypothesis, we conclude
    $(f \cdot x, a \cdot x \cdot b, f' \cdot x)
    \in \gamma_{\mathsf{B}}$. For any
    $(f',b,f') \in \rho_{\mathsf{A} \Rightarrow
      \mathsf{B}}$, since
    $b \sqsubseteq \lceil f' \rceil$, it follows
    from
    $(f,a \otimes \lceil f' \rceil,f') \in
    \rho_{\mathsf{A} \Rightarrow \mathsf{B}}$ that
    $(f,a \otimes b,f') \in \rho_{\mathsf{A}
      \Rightarrow \mathsf{B}}$. By
    left-quasi-reflexivity, we obtain
    $(f,a \otimes b,f) \in \rho_{\mathsf{A}
      \Rightarrow \mathsf{B}}$.
  \end{proof}
}

Let $q^{l}_{{\mathsf{A}}}$ indicate the
quasi-metric representing the left observational
quasi-metrics associated with the \mmm{}
$\rho_{\mathsf A}$.

\begin{proposition}\label{prop:rho-gamma}
  For every type $\mathsf{A}$, we have
  $q^{l}_{{\mathsf{A}}} = \gamma_{\mathsf{A}}$.
\end{proposition}
\longv{
  \begin{proof}
    We first show that $\gamma_{\mathsf{A}}$ is a
    subset of $q^{l}_{\mathsf{A}}$ by induction
    on $\mathsf{A}$. It is straightforward to
    check the case $\real$ and the case
    $\mathsf{A} \times \mathsf{B}$. We check the
    case $\mathsf{A} \Rightarrow \mathsf{B}$. Let
    $(f,a,f')$ be an element of
    $\gamma_{\mathsf{A} \Rightarrow \mathsf{B}}$,
    and let $(f',a',f'')$ be an element of
    $\rho_{\mathsf{A} \Rightarrow \mathsf{B}}$. We
    show that $(f,a \otimes a',f'')$ is an element of
    $\rho_{\mathsf{A} \Rightarrow \mathsf{B}}$.
    For any $(x,b,x') \in \rho_{\mathsf{A}}$,
    since $(x,b,x) \in \rho_{\mathsf{A}}$, we have
    \begin{align*}
      (f \cdot x,
      a \cdot x \cdot b,
      f' \cdot x)
      &\in \gamma_{\mathsf{B}},
      \\
      (f' \cdot x,
      a' \cdot x \cdot b,
      f'' \cdot x')
      &\in \rho_{\mathsf{B}}.
    \end{align*}
    Hence, by the induction hypothesis, we see
    that
    $(f \cdot x, (a \otimes a') \cdot x \cdot b,
    f'' \cdot x')$ is an element of
    $\rho_{\mathsf{B}}$. It remains to check that
    $(f \cdot x, (a \otimes a') \cdot x \cdot b, f
    \cdot x')$ is an element of
    $\rho_{\mathsf{B}}$. Since
    $(f',a',f'') \in \rho_{\mathsf{A} \Rightarrow
      \mathsf{B}}$, we have
    $(f',a',f') \in \rho_{\mathsf{A} \Rightarrow
      \mathsf{B}}$. Then, by the definition of
    $\gamma_{\mathsf{A} \Rightarrow \mathsf{B}}$,
    we obtain
    $(f,a \otimes a',f) \in \rho_{\mathsf{A}
      \Rightarrow \mathsf{B}}$. Hence,
    $(f \cdot x,(a \otimes a') \cdot x \cdot b, f
    \cdot x')$ is an element of
    $\rho_{\mathsf{B}}$. We next show that
    $q^{l}_{\mathsf{A}}$ is a subset of
    $\gamma_{\mathsf{A}}$. Again, it is
    straightforward to check the case $\real$ and
    the case $\mathsf{A} \times \mathsf{B}$. We
    check the case
    $\mathsf{A} \Rightarrow \mathsf{B}$. Let
    $(f,a,f')$ be an element of
    $q^{l}_{\mathsf{A} \Rightarrow
      \mathsf{B}}$, and let $(x,b,x')$ be an
    element of $\rho_{\mathsf{A}}$. Since
    $(f,a \otimes \lceil f' \rceil,f') \in
    \rho_{\mathsf{A} \Rightarrow \mathsf{B}}$ and
    $(x,\lceil x \rceil,x) \in \rho_{\mathsf{A}}$,
    we obtain
    \begin{equation*}
      (f \cdot x,
      (a \cdot x \cdot \lceil x \rceil)
      \otimes
      (\lceil f' \rceil \cdot x \cdot
      \lceil x \rceil),f' \cdot x)
      =
      (f \cdot x,
      (a \cdot x \cdot \lceil x \rceil)
      \otimes
      \lceil f' \cdot x \rceil ,
      f' \cdot x)
      \in \rho_{\mathsf{B}}
    \end{equation*}
    Hence, by Lemma~\ref{lem:rho-gamma},
    $(f \cdot x,a \cdot x \cdot \lceil x \rceil,f'
    \cdot x)$ is an element of
    $\rho_{\mathsf{B}}$. Since $a$ is monotone,
    for any $(x,b,x') \in \rho_{\mathsf{A}}$, we
    have
    $(f \cdot x,a \cdot x \cdot b,f' \cdot x) \in
    \rho_{\mathsf{B}}$. If
    $(f',b,f') \in \rho_{\mathsf{A} \Rightarrow
      \mathsf{B}}$, then by the definition of
    $q^{l}_{\mathsf{A} \Rightarrow \mathsf{B}}$,
    we obtain
    $(f,a \otimes b,f) \in \rho_{\mathsf{A}
      \Rightarrow \mathsf{B}}$.
  \end{proof}
}

We can use Proposition~\ref{prop:rho-gamma} to
over-approximate $\rho$-distances in terms of
$\gamma$-distances and the left observational
quasi-metric. Let us sketch our idea.
First, thanks to Proposition \ref{prop:rho-gamma}
and Proposition \ref{prop:left-right}, we can exploit
left-transitivity to pass from a $\gamma$-distance
between $\mathsf t$ and $\mathsf s$ and a
self-$\rho$-distance of $\mathsf s$
to a $\rho$-distance between
$\mathsf{t}$ and $\mathsf{s}$:
\begin{align}
  (\psem{\mathsf t}, a, \psem{\mathsf s})
  \in \gamma_{\mathsf A}
  \ \text{and} \ 
  (\psem{\mathsf s},b,\psem{\mathsf s})
  \in \rho_{{\mathsf A}}
  \ \implies \ 
  (\psem{\mathsf t}, a \otimes b,\psem{\mathsf s})
  \in \rho_{{\mathsf A}}.
  \tag{left transitivity}
\end{align} 
Second, thanks to the fundamental lemma, we can
always obtain a $\rho$-distance by summing a
$\gamma$-distance with the self-distance
$\fsem{\mathsf s}$:
\begin{align}
  (\psem{\mathsf t}, a, \psem{\mathsf s})
  \in \gamma_{\mathsf A} 
  \ \implies \ 
  (\psem{\mathsf t}, a \otimes \fsem{\mathsf s},
  \psem{\mathsf s}) \in \rho_{{\mathsf A}}.
  \tag{$\rho\sqsupseteq \gamma\ \otimes$ self-$\rho$}
\end{align} 
The following result exploits this last idea to
bound the distance between two functions $f$ and
$g$ by summing the ``vertical distance'' between
$f$ and $g$ (that is, the distance of $f(x)$ and
$g(x)$ for some fixed $x$) with an approximation
of the self-distances of $f$ and $g$:
%


\begin{theorem}\label{thm:approx}
  Let $\mathsf{A}$ be a type. For any
  $f,f' \in |\mathsf{A} \Rightarrow \real|$ and
  any
  $a,a' \in \mathcal {Q}_{\mathsf{A}
    \Rightarrow \real}$, if
  \begin{itemize}
  \item $|f \cdot x - f' \cdot x|
    \sqsupseteq a \cdot x \cdot b$ for all
    $(x,b,x') \in \rho_{\mathsf{A}}$; and
  \item
    $(f,a',f) \in \rho_{\mathsf{A} \Rightarrow
      \real}$ and
    $(f',a',f') \in \rho_{\mathsf{A} \Rightarrow
      \real}$,
  \end{itemize}
  then
  $(f, a \otimes a', f') \in \rho_{\mathsf{A}
    \Rightarrow \real}$.
\end{theorem}
\longv{
  \begin{proof}
    By the definition of
    $\gamma_{\mathsf{A} \Rightarrow \real}$, we
    obtain
    \begin{equation*}
      (f,a \otimes
      (\lceil f' \rceil \multimap \lceil f \rceil),
      f')
      \in \gamma_{\mathsf{A} \Rightarrow \real}.
    \end{equation*}
    Therefore, it follows from
    Proposition~\ref{prop:rho-gamma} that
    \begin{equation*}
      (f,a \otimes
      (\lceil f' \rceil \multimap \lceil f \rceil)
      \otimes \lceil f' \rceil,f')
      \in \rho_{\mathsf{A} \Rightarrow \real}.
    \end{equation*}
    Since
    $(f,a',f) \in \rho_{\mathsf{A} \Rightarrow
      \real}$ and
    $(f',a',f') \in \rho_{\mathsf{A} \Rightarrow
      \real}$, we have
    $a' \sqsubseteq (\lceil f' \rceil \multimap
    \lceil f \rceil) \otimes \lceil f' \rceil$.
    Hence,
    $(f, a \otimes a', f') \in \rho_{\mathsf{A}
      \Rightarrow \real}$.
  \end{proof}
}

\begin{example}\label{eg:id-sin}
  Let $\mathrm{id}_{\mathbb{R}}$ be the identity
  function on $\mathbb{R}$. By the fundamental
  lemma with a simple calculation, we obtain
  $(\mathrm{id}_{\mathbb{R}},a',
  \mathrm{id}_{\mathbb{R}}) \in \rho_{R
    \Rightarrow R}$ and
  $(\sin,a',\sin) \in \rho_{R \Rightarrow R}$
  where $a' \cdot x \cdot b = b$. By
  Theorem~\ref{thm:approx},
  $a \in \mathcal{Q}_{R \Rightarrow R}$ given by
  $a \cdot x \cdot b = |x - \sin(x)|$
  satisfies
  $(\mathrm{id}_{\mathbb{R}},a \otimes a', \sin)
  \in \rho_{R \Rightarrow R}$. To be concrete,
  $(a \otimes a') \cdot x \cdot b = |x - \sin(x)|
  + b$, which means that the distance between
  $x$ and $\sin(y)$ is
  small when $x$ and $y$ are close to $0$.
\end{example}

\begin{remark}
  Due to asymmetry in the definition of the
  exponential $X \Rightarrow Y$ in $\qqm$, it is
  not clear how to capture the \emph{right}
  observational quasi-metrics in a similar manner.
  However, we will see that right observational
  quasi-metrics can be captured by partial metric
  logical relations introduced in the next
  section.
\end{remark}



\section{Partial Metric Logical Relations}
\label{sec:diff-logic-relat}

As discussed in Section 3, the quasi$^2$-metrics $\rho_{\mathsf A}$ are not, in general, partial metrics. 
In this section we introduce a family of differential logical relations $(\eta_{\mathsf A})_{\mathsf A\in \mathsf{Types}}$ that defines a class of partial quasi-metrics over $\LL$.
The fundamental (indeed, the only) difference with respect to the family $\rho_{\mathsf A}$ is, as it may be expected, in the case of the function type.

For any type $\mathsf A$, we define $\eta_{\mathsf A}\subseteq |\mathsf A|\times \mathcal{Q}_{\mathsf A}\times |\mathsf A|$ by induction on $\mathsf A$ as follows:
\begin{align*}
  (x,a,x') \in \eta_{\real}
  &\iff 
    |x-x'|\sqsupseteq a,
  \\
  (f,a,f') \in \eta_{\mathsf A \Rightarrow \mathsf B}
  &\iff
    \text{there are }  a_1, a_{2}\in \mathcal{Q}_{\mathsf A \Rightarrow \mathsf B}
    \text{ such that }
    a_{1} \otimes a_{2}
    \sqsupseteq a
    \text{ and } \\
  &\phantom{{} \iff {}}
    \text{for all }
    (x, b,x') \in \eta_{\mathsf A}, \ 
    (f\cdot x, a_{1} \cdot x \cdot  b,f\cdot x')
    \in \eta_{\mathsf B}
    \text{ and }\\
  &\phantom{{} \iff {}
    \text{for all }
    (x, b,x') \in \eta_{\mathsf A}\ }
    (f\cdot x', a_{2} \cdot x \cdot b,f'\cdot x')
    \in \eta_{\mathsf B},\\
  ((x,y), (a,b),(x',y')) \in
  \eta_{\mathsf A \times \mathsf B}
  &\iff
    (x,a,x')\in  \eta_{\mathsf A}
    \text{ and }
    (y,b,y')\in       \eta_{\mathsf B}.
\end{align*}
The idea of the definition of
$\eta_{\mathsf A \Rightarrow\mathsf B}$ is that if
$(f,a,f') \in \eta_{\mathsf A \Rightarrow
  \mathsf B}$, then $a$ must be larger than or equal to the sum of 
the self-distance of $f$ and of the ``vertical''
distances between $f$ and $f'$.
\shortv{The following result shows that the relations $\eta_{\mathsf A}$ define partial quasi-metrics on all types.}

\longv{
For
$(x,a,x') \in \eta_{\mathsf A \Rightarrow
  \mathsf B}$, we call a pair
$ a_{1}, a_{2} \in
\mathcal{Q}_{\mathsf A \Rightarrow \mathsf B}$
satisfying the condition in the definition of
$(x, a,x') \in \eta_{\mathsf A \Rightarrow
  \mathsf B}$ a \emph{decomposition of
  $(x, a,x')\in \eta_{\mathsf A \Rightarrow
  \mathsf B}$}.
}

\longv{
The following result shows that the relations $\eta_{\mathsf A}$ are $\mathcal Q$-closed.
\begin{lemma}\label{lem:inf}
For any type $\mathsf A $, the relation $\rho_{\mathsf A}$ is $\mathcal Q$-closed.
\end{lemma}
}
\longv{An immediate consequence of the lemma is the following:
\begin{corollary}
For all $(x, a, x')\in \eta_{\mathsf  A}$, the set of decompositions of $(x, a, x')\in \eta_{\mathsf  A}$ is a complete lattice.
\end{corollary}
}
\longv{The following result shows that the relations $\eta_{\mathsf A}$ define partial quasi-metrics on all types.}

\begin{proposition}
  For all types $\mathsf A$:
  \begin{itemize}
  \item If
    $(x, a,x') \in
    \eta_{\mathsf A}$, then
    $(x, a,x) \in
    \eta_{\mathsf A}$.
  \item If $(x, a,z) \in \eta_{\mathsf A}$ and
    $(z, b,y) \in \eta_{\mathsf A}$, then there
    exists $c_1, c_2 \in \mathcal{Q}_{\mathsf{A}}$
    such that
    $a \otimes b \sqsubseteq c_{1} \otimes c_{2}$,
    $(z, c_{1}, z) \in \eta_{\mathsf A}$ and
    $(x,c_{2},y) \in \eta_{\mathsf{A}}$. In
    particular,
    $(x, a \otimes (c_{1} \multimap b), y) \in
    \eta_{\mathsf A}$.
%
  \end{itemize}
\end{proposition}
\longv{
  \begin{proof}
    We only prove the second, more delicate,
    statement. By induction on $\mathsf A$, we
    show that there is a map
    $\varphi_{\mathsf A} \colon
    \mathcal{Q}_{\mathsf{A}} \times
    \mathcal{Q}_{\mathsf{A}} \Rightarrow
    \mathcal{Q}_{\mathsf{A}} \times
    \mathcal{Q}_{\mathsf{A}}$ such that if
    $(x,a,z) \in \eta_{\mathsf A}$ and
    $(z,b,y) \in \eta_{\mathsf A}$, then
    $(c_1,c_2) = \varphi_{\mathsf A}(a, b)$
    satisfies the required conditions. For the
    base case, we define
    $ \varphi_{\real}(a, b) = (0, a + b)$. For the
    case
    $\mathsf A = (\mathsf{B} \Rightarrow \mathsf{C})$,
    let
    $(a_{1},a_{2}) \in \mathcal{Q}_{\mathsf{B}
      \Rightarrow \mathsf{C}} \times
    \mathcal{Q}_{\mathsf{B} \Rightarrow
      \mathsf{C}}$ and
    $(b_{1},b_{2}) \in \mathcal{Q}_{\mathsf{B}
      \Rightarrow \mathsf{C}} \times
    \mathcal{Q}_{\mathsf{B} \Rightarrow
      \mathsf{C}}$ be the greatest decompositions
    of $a$ and $b$, respectively. We define
    $\varphi_{\mathsf{B} \Rightarrow
      \mathsf{C}}(a,b)$ by
    \begin{equation*}
      \varphi_{\mathsf{B} \Rightarrow
      \mathsf{C}}(a,b)
      = (a_{1}
      \otimes k,
      b_{1}
      \otimes
      l)
    \end{equation*}
    where
    $(k \cdot w \cdot d,l \cdot w \cdot d) =
    \varphi_{\mathsf C}
    (a_{2} \cdot w \cdot d,b_{2} \cdot w \cdot d)$.
    Below, we write
    $(c_1,c_2)$ for $\varphi_{\mathsf{B}
    \Rightarrow \mathsf{C}}(a,b)$.
    Let us
    check that $\varphi_{\mathsf B \Rightarrow \mathsf C}$
    is a witness.
    \begin{itemize}
    \item We first show that
      $(z,c_1,z)\in \eta_{\mathsf B \Rightarrow
        \mathsf C}$. For any
      $(w,d,w')\in \eta_{\mathsf B}$, we have
      \begin{align}
        &(x \cdot w',
          a_{2}\cdot w \cdot d,
          z \cdot w')\in \eta_{\mathsf{C}},
          \label{aeq:nn1} \\
        &(z \cdot w,
          b_{1}\cdot w \cdot d,
          z \cdot w')\in \eta_{\mathsf{C}},
          \label{aeq:nn2} \\
        &(z \cdot w',
          b_{2}\cdot w \cdot d,
          y \cdot w')\in \eta_{\mathsf{C}}.
          \label{aeq:nn3}
      \end{align}
      Then, by applying the induction hypothesis
      to \eqref{aeq:nn1} and \eqref{aeq:nn3}, we
      obtain
      \begin{equation}
        (z \cdot w',
        k \cdot w \cdot d,
        z \cdot w')
        \in \eta_{\mathsf C}.
        \label{aeq:nn4}
      \end{equation}
      By \eqref{aeq:nn2} and \eqref{aeq:nn4}, we
      obtain
      $(z,c_{1},z) \in \eta_{\mathsf{B}
        \Rightarrow \mathsf{C}}$.
    \item We next show that
      $(x,c_2 ,y) \in \eta_{\mathsf{B} \Rightarrow
        \mathsf{C}}$. For any
      $(w,d,w')\in
      \eta_{\mathsf B}$, we have
      \begin{align}
        &(x \cdot w, a_{1} \cdot w \cdot d,
          x \cdot w') \in \eta_{\mathsf C},
          \label{aeq:ml1} \\
        &(x \cdot w',
          a_{2} \cdot w \cdot d,
          z \cdot w') \in \eta_{\mathsf C},
          \label{aeq:ml2} \\
        &(z \cdot w',
          b_{2} \cdot w \cdot d,
          y \cdot w')
          \in \eta_{\mathsf C}.
          \label{aeq:ml3}
      \end{align}
      By applying the induction hypothesis to
      \eqref{aeq:ml2} and \eqref{aeq:ml3}, we
      obtain
      \begin{equation}
        (x \cdot  w',
        l \cdot w \cdot d,
        y \cdot w') \in \eta_{\mathsf C}.
        \label{aeq:ml4}
      \end{equation}
      By \eqref{aeq:ml1} and \eqref{aeq:ml4},
      we obtain
      $(x,c_{2},y) \in \eta_{\mathsf{B}
      \Rightarrow \mathsf{C}}$.
    \item Finally, we have
      \begin{align*}
        (c_{1} \cdot w \cdot d)
        \otimes
        (c_{2} \cdot w \cdot d)
        &=
          (a_{1} \cdot w \cdot d)
          \otimes
          (k \cdot w \cdot d)
          \otimes
          (b_{1} \cdot w \cdot d)
          \otimes
          (l \cdot w \cdot d)
        \\
        &\sqsupseteq
          (a_{1} \cdot w \cdot d)
          \otimes
          (a_{2} \cdot w \cdot d)
          \otimes
          (b_{1} \cdot w \cdot d)
          \otimes
          (b_{2} \cdot w \cdot e)
        \\
        &\sqsupseteq
          (a \cdot w \cdot d)
          \otimes
          (b \cdot w \cdot d).
      \end{align*}
    \end{itemize}
  \end{proof}
}

By adapting the definition of $\gamma_{\mathsf A}$ from Section 6, we can capture the left observational quasi-metrics $q_{\mathsf A}^l$ associated with the partial quasi-metrics $\eta_{\mathsf A}$.
Moreover, by Proposition 
\ref{prop:quasipartial}, the \emph{right} observational quasi-metric $q_{\mathsf A}^r$ satisfies $
(x,\widehat{\eta}_{\mathsf A}(x,x)\multimap a,y)\in q_{\mathsf A}^r \ \iff (x,a,y)\in \eta_{\mathsf A}$.
%
Thanks to this, we can capture this quasi-metrics via the logical relations $\delta_{\mathsf A}\subseteq |\mathsf A|\times \mathcal{Q}_{\mathsf A}\times |\mathsf A|$ defined by induction on $\mathsf A$, letting the base and product case being defined as for $\gamma_{\mathsf A}$, and the function case being as follows:
\begin{align*}
  (f,a,f') \in
  \delta_{\mathsf{A} \Rightarrow \mathsf{B}}
  &\iff \text{for all }
    (f,b,f) \in
    \eta_{\mathsf{A}\Rightarrow \mathsf{B}},\
    (f,a \otimes b, f') \in
    \eta_{\mathsf{A} \Rightarrow \mathsf{B}}.
\end{align*}

\begin{proposition}
For every type $\mathsf A$, we have $q_{\mathsf A}^r=\delta_{\mathsf A}$.
\end{proposition}
\longv{\begin{proof}
The only interesting case is that of a function type $\mathsf A=\mathsf B\Rightarrow \mathsf C$.
By Proposition 
\ref{prop:quasipartial}, $(f,a,f')\in q_{\mathsf A}^r$ holds iff for all $a\sqsubseteq \widehat{q}_{\mathsf A}^r(f,f')=\widehat{q}_{\mathsf A}^r(f,f)\multimap \widehat{q}_{\mathsf A}^r(f,f')$, which is in turn equivalent to $a\otimes \widehat{q}_{\mathsf A}^r(f,f)\sqsubseteq \widehat{q}_{\mathsf A}^r (f,f')$. 
This implies then that $(f,a,f')\in q_{\mathsf A}^r$ iff for all $(f,b,f)\in q_{\mathsf A}^r$ (i.e.~for all $b\sqsubseteq \widehat{q}^r_{\mathsf A}(f,f)$), $(f,a\otimes b,f')\in q_{\mathsf A}^r$ (i.e.~$a\otimes b\sqsubseteq \widehat{q}^r_{\mathsf A}(f,f')$), that is, iff $(f,a,f')\in \delta_{\mathsf A}$.
\end{proof}
}


\section{A Quantitative Equational Theory}
\label{sec:synt-constr-mmm}

The goal of this section is to introduce an
equational theory to formally deduce differences
between programs. To this end, we first give a
syntactic presentation of differential logical
relations internally to the language of $\LL$, and
then introduce a deductive system to deduce
program differences.

While our idea is inspired by the quantitative
equational theories of Mardare et
al.~\cite{Plotk}, it differs in two respects:
first, distances need not be real numbers, but are
presented as arbitrary $\LL$-programs; second,
non-expansiveness is replaced by the condition
corresponding to the fundamental lemma of
differential logical relations.

\subparagraph{Preparation}

Before we go into construction, we prepare some
syntactic counter parts of constructions
in the fundamental lemma for $\qqm$. We first
inductively define a type $\prt{\mathsf{A}}$ by
\begin{equation*}
  \prt{\real} = \real,
  \qquad
  \prt{(\mathsf{A} \Rightarrow \mathsf{B})}
  = \mathsf{A} \Rightarrow
  \prt{\mathsf{A}} \Rightarrow
  \prt{\mathsf{B}},
  \qquad
  \prt{(\mathsf{A} \times \mathsf{B})}
  = \prt{\mathsf{A}} \times \prt{\mathsf{B}}.
\end{equation*}
This is
a syntactic counter part of quantales
$\mathcal{Q}_{\mathsf{A}}$. The reason that we
define $\prt{\real}$ to be $\real$ even though
$\prt{\real}$ should be a type of non-negative
extended real numbers is to keep the syntax of
$\LL$ simple. It is possible to extend $\LL$ with
a type $\real_{\geq 0}^{\infty}$ of non-negative
extended real numbers and types
$\mathsf{A} \rightarrowtriangle \mathsf{B}$ of
monotone functions. We next give syntactic counter
part of $\fsem{\mathsf{t}}$. For this purpose, we
suppose that there is a partition
$\mathsf{Var} = \mathsf{Var}_{0} \cup
\mathsf{Var}_{1}$, i.e., there are mutually
disjoint subsets
$\mathsf{Var}_{0},\mathsf{Var}_{1} \subseteq
\mathsf{Var}$ such that $\mathsf{Var}$ is equal to
$\mathsf{Var}_{0} \cup \mathsf{Var}_{1}$.
Furthermore, we suppose that there is a bijection
$\dot{(-)} \colon \mathsf{Var}_{0} \to
\mathsf{Var}_{1}$. In the sequel, we denote
variables in $\mathsf{Var}_{1}$ by dotted symbols
$\dot{\mathsf{x}}, \dot{\mathsf{y}},
\dot{\mathsf{z}},\ldots$, and we denote variables
in $\mathsf{Var}_{0}$ by
$\mathsf{x},\mathsf{y},\mathsf{z},\ldots$. Based
on this convention, for a typing context
$\mathsf{\Gamma} = (\mathsf{x}_{1}:\mathsf{A}_{1},
\ldots, \mathsf{x}_{n}:\mathsf{A}_{n})$, we define
a typing context $\prt{\mathsf{\Gamma}}$ by
$\prt{\mathsf{\Gamma}} =
(\dot{\mathsf{x}}_{1}:\prt{\mathsf{A}}_{1},
\ldots,
\dot{\mathsf{x}}_{n}:\prt{\mathsf{A}}_{n})$. Now,
for a term
$\mathsf{\Gamma} \vdash \mathsf{t} : \mathsf{A}$,
we define a term
$\mathsf{\Gamma},\prt{\mathsf{\Gamma}} \vdash
\prt{\mathsf{t}}:\mathsf{A}$, which we call the
\emph{derivative} of $\mathsf{t}$, in
Figure~\ref{fig:diff-term}. The definition of
$\prt{\mathsf{t}}$ corresponds to the definition
of $\fsem{\mathsf{t}}$, and we can find the same
construction in \cite{DLG21}.



\begin{figure}
  \fbox{\begin{minipage}{.96\textwidth}
      \footnotesize \centering
      \begin{equation*}
        \begin{array}{c}
          \prt{\mathsf{x}}
          = \dot{\mathsf{x}}
          \qquad
          \prt{\underline{r}}
          = \underline{0}
          \qquad
          \prt{(\mathsf{t} \, \mathsf{s})}
          =
          \prt{\mathsf{t}} \,
         \mathsf{s}\,
         \prt{\mathsf{s}}
          \qquad
          \prt{(\lam{\mathsf{x}}{
          \mathsf{A}}{\mathsf{t}})}
          = 
          {
          \lam{\mathsf{x}}{\mathsf{A}}
          \lam{\dot{\mathsf{x}}}{
          \prt{\mathsf{A}}}
          {
          \prt{\mathsf{t}}}}
          \qquad
          \prt{\pair{\mathsf{t}}{\mathsf{s}}}
          = \pair{\prt{\mathsf{t}}}{\prt{\mathsf{s}}}
          \\[7pt]
          \prt{(\fst(\mathsf{t}))}
          = \fst(\prt{\mathsf{t}})
          \qquad
          \prt{(\snd(\mathsf{t}))}
          = \snd(\prt{\mathsf{t}})
          \qquad
          \prt{(\phi
          (\mathsf{t}_{1},\ldots,\mathsf{t}_{n}))}
          = \phi^{\mathrm{d}}
          (\mathsf{t}_{1},
          \ldots,
          \mathsf{t}_{n},
          \prt{\mathsf{t}_{1}},
          \ldots,
          \prt{\mathsf{t}_{n}})
        \end{array}
      \end{equation*}
    \end{minipage}}
  \caption{Derivative of Term}
  \label{fig:diff-term}
\end{figure}

\subparagraph*{Syntactic Differential Logical Relations}
\label{sec:defin-diff-logic}

By adopting the structure of $\qqm$, we
define a type-indexed family
$\{\dlog_{\mathsf{A}} \subseteq \cterm{\mathsf{A}}
\times \cterm{\prt{\mathsf{A}}} \times
\cterm{\mathsf{A}}\}_{\mathsf{A} \in \type}$ of
ternary predicates as follows:
\begin{align*}
  (\mathsf{t},\mathsf{a},\mathsf{t}')
  \in \dlog_{\real}
  &\iff
    \text{there are }
    r,r' \in \mathbb{R}
    \text{ and }
    s \in [0,+\infty]
    \text{ such that }
    |r - r'| \sqsupseteq s
    \text{ and } \\
  &\mathrel{\phantom{\iff}}
    \vdash \mathsf{t} = \underline{r}:\real
    \text{ and }
    \vdash \mathsf{a} = \underline{s}:\real
    \text{ and }
    \vdash \mathsf{t}' = \underline{r}':\real
    , \\
  (\mathsf{t},\mathsf{a},\mathsf{t}')
  \in \dlog_{\mathsf{A}
  \Rightarrow \mathsf{B}}
  &\iff
    \text{for any }
    (\mathsf{s},\mathsf{b},\mathsf{s}')
    \in \dlog_{\mathsf{A}},
    (
    \mathsf{t} \, \mathsf{s},
    \mathsf{a} \, \mathsf{s} \, \mathsf{b},
    \mathsf{t}' \, \mathsf{s}'
    )
    \in \dlog_{\mathsf{B}}
    \text{ and }
    (
    \mathsf{t} \, \mathsf{s},
    \mathsf{a} \, \mathsf{s} \, \mathsf{b},
    \mathsf{t} \, \mathsf{s}'
    )
    \in \dlog_{\mathsf{B}}
    , \\
  (\mathsf{t},\mathsf{a},\mathsf{t}')
  \in \dlog_{\mathsf{A} \times \mathsf{B}}
  &\iff
    (
    \fst(\mathsf{t}),
    \fst(\mathsf{a}),
    \fst(\mathsf{t}')
    )
    \in \dlog_{\mathsf{A}}
    \text{ and }
    (
    \snd(\mathsf{t}),
    \snd(\mathsf{a}),
    \snd(\mathsf{t}'))
    \in \dlog_{\mathsf{B}}
\end{align*}
where we write
$\mathsf{\Gamma} \vdash \mathsf{t} = \mathsf{s}:
\mathsf{A}$ when the equality between
$\mathsf{\Gamma} \vdash \mathsf{t} :\mathsf{A}$
and
$\mathsf{\Gamma} \vdash \mathsf{s} :\mathsf{A}$ is
derivable from the standard equational theory
consisting of $\beta\eta$-equalities extended with
the following axiom for every multi-arity function
$\phi$:
\begin{equation*}
    \mathsf{\Gamma} \vdash
    \phi(\underline{r_{1}},
    \ldots,\underline{r_{\mathrm{ar}(\phi)}})
    =
    \underline{\phi(r_{1},\ldots,r_{\mathrm{ar}(\phi)})}
    : \real.
\end{equation*}
Although $\cterm{\mathsf{A}}$ is not a quantale in
general, we can show that $\dlog_{\mathsf{A}}$
satisfies ``left quasi-reflexivity'',
``transitivity'' and a fundamental lemma
in the following form.
\begin{proposition}\label{lem:dlog}
  Let $\mathsf{A}$ be a type.
  \begin{itemize}
  \item If
    $(\mathsf{t},\mathsf{a},\mathsf{t}') \in
    \dlog_{\mathsf{A}}$, then
    $(\mathsf{t},\mathsf{a},\mathsf{t}) \in
    \dlog_{\mathsf{A}}$.
  \item If
    $(\mathsf{t},\mathsf{a},\mathsf{t}') \in
    \dlog_{\mathsf{A}}$ and
    $(\mathsf{t}',\mathsf{a'},\mathsf{t}'') \in
    \dlog_{\mathsf{A}}$, then
    $(\mathsf{t},\mathsf{add}_{\mathsf{A}} \,
    \mathsf{a} \, \mathsf{a'},\mathsf{t}'') \in
    \dlog_{\mathsf{A}}$ where
    $\mathsf{add}_{\mathsf{A}} \in
    \cterm{\mathsf{A} \Rightarrow \mathsf{A}
      \Rightarrow \mathsf{A}}$ \shortv{is the
      pointwise extension of the standard addition
      on $\real$.} \longv{is given by
      \begin{align*}
        \mathsf{add}_{\real}
        &= \lam{\mathsf{x}\mathsf{y}}{\real}
          {\underline{\mathrm{add}}
          (\mathsf{x}, \mathsf{y})},
        \\
        \mathsf{add}_{\mathsf{A} \Rightarrow \mathsf{B}}
        &=
          \lam{\mathsf{x}\mathsf{y}}
          {\mathsf{A} \Rightarrow \mathsf{B}}
          {\lam{\mathsf{z}}{\mathsf{A}}
          {\mathsf{add}_{\mathsf{B}} \,
          (\mathsf{x} \, \mathsf{z}) \,
          (\mathsf{y} \, \mathsf{z})}},
        \\
        \mathsf{add}_{\mathsf{A} \times \mathsf{B}}
        &=
          \lam{\mathsf{x}\mathsf{y}}{\mathsf{A} \times \mathsf{B}}
          {\pair
          {\mathsf{add}_{\mathsf{A}} \, \fst(\mathsf{x}) \, \fst(\mathsf{y})}
          {\mathsf{add}_{\mathsf{B}} \, \snd(\mathsf{x}) \, \snd(\mathsf{y})}
          }.
      \end{align*}}
  \item For any term $\mathsf{x}_{1}:\mathsf{A}_{1},\ldots,\mathsf{x}_{n}:\mathsf{A}_{n}
    \vdash \mathsf{t}:\mathsf{A}$, and for any family
    $\{(\mathsf{s}_{i},\mathsf{a}_{i},\mathsf{s}'_{i}) \in \dlog_{\mathsf{A}_{i}}\}_{1 \leq i \leq n}$,
    \begin{equation*}
      (\mathsf{t}[\mathsf{s}_{1}/\mathsf{x}_{1},\ldots,\mathsf{s}_{n}/\mathsf{x}_{n}],
      \prt{\mathsf{t}}[\mathsf{s}_{1}/\mathsf{x}_{1},\ldots,\mathsf{s}_{n}/\mathsf{x}_{n},
      \mathsf{a}_{1}/\dot{\mathsf{x}}_{1},\ldots,\mathsf{a}_{n}/\dot{\mathsf{x}}_{n}],
      \mathsf{t}'[\mathsf{s}'_{1}/\mathsf{x}_{1},\ldots,\mathsf{s}'_{n}/\mathsf{x}_{n}]
      ) \in \dlog_{\mathsf{A}}.
    \end{equation*}
    \longv{\item If
      $(\mathsf{t}, \mathsf{a}, \mathsf{t'}) \in
      \dlog_{\mathsf{A}}$ and
      $\vdash \mathsf{t} = \mathsf{s} :
      \mathsf{A}$ and
      $\vdash \mathsf{a} = \mathsf{b} :
      \prt{\mathsf{A}}$ and
      $\vdash \mathsf{t}' = \mathsf{s}' :
      \mathsf{A}$, then
      $(\mathsf{s}, \mathsf{b}, \mathsf{s}') \in
      \dlog_{\mathsf{A}}$.}
  \end{itemize}
\end{proposition}
\longv{
  \begin{proof}[Proof Sketch]
    We prove the statement by induction on
    $\mathsf{A}$. We only check the case
    $\mathsf{A} \Rightarrow \mathsf{B}$. It is
    straightforward to derive
    ``left-quasi-reflexivity'' from the definition
    of
    $\dlog_{\mathsf{A} \Rightarrow \mathsf{B}}$.
    This is why we modify the definition of
    differential logical relation given in
    \cite{DGY19}. For transitivity, we shall show
    that for any
    $(\mathsf{t},\mathsf{a},\mathsf{t}') \in
    \dlog_{\mathsf{A}}$,
    $(\mathsf{t}',\mathsf{a}',\mathsf{t}'') \in
    \dlog_{\mathsf{A}}$ and
    $(\mathsf{s},\mathsf{b},\mathsf{s'}) \in
    \dlog_{\mathsf{A}}$, we have
    $(\mathsf{t} \, \mathsf{s}, (\mathsf{a} +
    \mathsf{a'}) \, \mathsf{b} \, \mathsf{s},
    \mathsf{t''} \, \mathsf{s'}) \in
    \dlog_{\mathsf{B}}$. By the induction
    hypothesis, we obtain
    $(\mathsf{s},\mathsf{b},\mathsf{s}) \in
    \dlog_{\mathsf{A}}$. Therefore,
    \begin{equation*}
      (\mathsf{t} \, \mathsf{s},
      \mathsf{a} \, \mathsf{b} \,
      \mathsf{s},
      \mathsf{t'} \, \mathsf{s}
      ) \in
      \dlog_{\mathsf{B}},
      \qquad
      (\mathsf{t'} \, \mathsf{s},
      \mathsf{a'} \, \mathsf{b} \,
      \mathsf{s},
      \mathsf{t''} \, \mathsf{s'}
      ) \in
      \dlog_{\mathsf{B}}.
    \end{equation*}
    Then, by transitivity of $\dlog_{\mathsf{B}}$,
    we obtain
    $(\mathsf{t} \, \mathsf{s}, (\mathsf{a} +
    \mathsf{a'}) \, \mathsf{b} \, \mathsf{s},
    \mathsf{t''} \, \mathsf{s'}) \in
    \dlog_{\mathsf{B}}$. We can prove the fundamental lemma
    by induction on the derivation of $\mathsf{\Gamma}
    \vdash \mathsf{t} : \mathsf{B}$.
  \end{proof}
}

\subparagraph{Equational Metric}
\label{sec:equat-theory-dist}

We introduce a formal system to infer
$\dlog$-distances between terms. For terms
$\mathsf{\Gamma} \vdash \mathsf{t} : \mathsf{A}$
and
$\mathsf{\Gamma},\prt{\mathsf{\Gamma}} \vdash
\mathsf{a} : \mathsf{A}$ and
$\mathsf{\Gamma} \vdash \mathsf{t}' : \mathsf{A}$,
we write
$\mathsf{\Gamma} \vdash
(\mathsf{t},\mathsf{a},\mathsf{t'}) : \mathsf{A}$
when we can derive this judgment from the rules
given in Figure~\ref{fig:eqt}. Then, we define a
type-indexed ternary predicates
$\{\deq_{\mathsf{A}} \subseteq \cterm{\mathsf{A}}
\times \cterm{\prt{\mathsf{A}}} \times
\cterm{\mathsf{A}}\}_{\mathsf{A} \in \type}$ by
\begin{equation*}
  (\mathsf{t},
  \mathsf{a},
  \mathsf{t'}) 
  \in \deq_{\mathsf{A}}
  \iff
  \vdash
  (\mathsf{t},
  \mathsf{a},
  \mathsf{t'}) : \mathsf{A}.
\end{equation*}
We note that quasi-reflexivity and transitivity for
arbitrary $\mathsf{A}$ follows from left
quasi-reflexivity and transitivity for $\real$. We can
also show that $\deq$ is subsumed by $\dlog$.
\begin{proposition}\label{prop:deq}
  Let $\mathsf{A}$ be a type.
  \begin{itemize}
  \item If
    $(\mathsf{t},\mathsf{a},\mathsf{t}') \in
    \deq_{\mathsf{A}}$, then
    $(\mathsf{t},\mathsf{a},\mathsf{t}) \in
    \deq_{\mathsf{A}}$.
  \item If
    $(\mathsf{t},\mathsf{a},\mathsf{t}') \in
    \deq_{\mathsf{A}}$ and
    $(\mathsf{t}',\mathsf{a'},\mathsf{t}'') \in
    \deq_{\mathsf{A}}$, then
    $( \mathsf{t},\mathsf{add}_{\mathsf{A}} \,
    \mathsf{a} \, \mathsf{a'},\mathsf{t}'') \in
    \deq_{\mathsf{A}}$.
  \item For any term
    $\mathsf{x}_{1}:\mathsf{A}_{1},\ldots,\mathsf{x}_{n}:\mathsf{A}_{n}
    \vdash \mathsf{t}:\mathsf{A}$, and for any
    family
    $\{(\mathsf{s}_{i},\mathsf{a}_{i},\mathsf{s}'_{i})
    \in \deq_{\mathsf{A}_{i}}\}_{1 \leq i \leq
      n}$,
    \begin{equation*}
      (\mathsf{t}[\mathsf{s}_{1}/\mathsf{x}_{1},\ldots,\mathsf{s}_{n}/\mathsf{x}_{n}],
      \prt{\mathsf{t}}[\mathsf{s}_{1}/\mathsf{x}_{1},\ldots,\mathsf{s}_{n}/\mathsf{x}_{n},
      \mathsf{a}_{1}/\dot{\mathsf{x}}_{1},\ldots,\mathsf{a}_{n}/\dot{\mathsf{x}}_{n}],
      \mathsf{t}'[\mathsf{s}'_{1}/\mathsf{x}_{1},\ldots,\mathsf{s}'_{n}/\mathsf{x}_{n}]
      ) \in \deq_{\mathsf{A}}.
    \end{equation*}
  \item If
    $(\mathsf{t},\mathsf{a},\mathsf{t}') \in
    \deq_{\mathsf{A}}$, then
    $(\mathsf{t},\mathsf{a},\mathsf{t}') \in
    \dlog_{\mathsf{A}}$.
  \end{itemize}
\end{proposition}
\longv{
  \begin{proof}
    By induction on $\mathsf{A}$.
  \end{proof}
}

\begin{figure}
  \fbox{
    \begin{minipage}{0.9668\textwidth}
      \footnotesize \centering
      \begin{math}
        \begin{array}{c}
          \prftree{
          |r - r'| \leq s
          }{
          \mathsf{\Gamma} \vdash
          (\underline{r},
          \underline{s},
          \underline{r}') : \real
          }
          \qquad
          \prftree{
          \mathsf{\Gamma} \vdash
          (\mathsf{t}_{1},
          \mathsf{a}_{1},
          \mathsf{t'}_{1}) : \real
          }{
          \ldots
          }{
          \mathsf{\Gamma} \vdash
          (\mathsf{t}_{n},
          \mathsf{a}_{n},
          \mathsf{t'}_{n}) : \real
          }{
          \mathsf{\Gamma} \vdash
          (
          \phi
          (\mathsf{t}_{1},
          \ldots,
          \mathsf{t}_{n}),
          \phi^{\mathrm{d}}
          (
          \mathsf{t}_{1},
          \ldots,
          \mathsf{t}_{n},
          \mathsf{a}_{1},
          \ldots,
          \mathsf{a}_{n}
          ),
          \phi
          (\mathsf{t'}_{1},
          \ldots,
          \mathsf{t'}_{n})
          ) : \real
          }
          \\[7pt]
          \prftree{
          \mathsf{x}: \mathsf{A}
          \in \mathsf{\Gamma}
          }{
          \mathsf{\Gamma} \vdash
          (
          \mathsf{x},
          \dot{\mathsf{x}},
          \mathsf{x})
          : \mathsf{A}
          }
          \qquad
          \prftree{
          \mathsf{\Gamma} \vdash
          (\mathsf{t},\mathsf{a},\mathsf{t}')
          : \real
          }{
          \mathsf{\Gamma} \vdash
          (\mathsf{t}',\mathsf{a}',\mathsf{t}'')
          : \real
          }{
          \mathsf{\Gamma} \vdash
          (\mathsf{t},
          \mathsf{a} + \mathsf{a}',\mathsf{t}'')
          : \real
          }
          \qquad
          \prftree{
          \mathsf{\Gamma} \vdash
          (\mathsf{t},\mathsf{a},\mathsf{t'})
          : \real
          }{
          \mathsf{\Gamma} \vdash
          (\mathsf{t},\mathsf{a},\mathsf{t})
          : \real
          }
          \\[7pt]
          \prftree{
          \mathsf{\Gamma} , \mathsf{x}:\mathsf{A}, \dot{\mathsf{x}}:\prt{\mathsf{A}}
          \vdash
          (\mathsf{t},\mathsf{a},\mathsf{t'})
          : \mathsf{B}
          }{
          \mathsf{\Gamma} \vdash
          (
          \lam{\mathsf{x}}{
          \mathsf{A}}{
          \mathsf{t}},
          \lam{\dot{\mathsf{x}}}{
          \prt{\mathsf{A}}}{
          \lam{\mathsf{x}}{
          \mathsf{A}}{
          \mathsf{a}}},
          \lam{\mathsf{x}}{
          \mathsf{A}}{
          \mathsf{t'}}
          ) :
          \mathsf{A} \Rightarrow \mathsf{B}
          }
          \qquad
          \prftree{
          \mathsf{\Gamma} \vdash
          (\mathsf{t},\mathsf{a},\mathsf{t'})
          : \mathsf{A} \Rightarrow \mathsf{B}
          }{
          \mathsf{\Gamma} \vdash
          (\mathsf{s},\mathsf{b},\mathsf{s'})
          : \mathsf{A}
          }{
          \mathsf{\Gamma} \vdash
          (\mathsf{t} \, \mathsf{s},
          \mathsf{a} \, \mathsf{s}  \, \mathsf{b},
          \mathsf{t'} \, \mathsf{s'})
          : \mathsf{B}
          }
          \\[7pt]
          \prftree{
          \mathsf{\Gamma} \vdash
          (\mathsf{t},
          \mathsf{a},
          \mathsf{t'})
          : \mathsf{A} \times \mathsf{B}
          }{
          \mathsf{\Gamma} \vdash
          (\fst(\mathsf{t}),
          \fst(\mathsf{a}),
          \fst(\mathsf{t'}))
          : \mathsf{A}
          }
          \quad
          \prftree{
          \mathsf{\Gamma} \vdash
          (\mathsf{t},
          \mathsf{a},
          \mathsf{t'})
          : \mathsf{A} \times \mathsf{B}
          }{
          \mathsf{\Gamma} \vdash
          (\snd(\mathsf{t}),
          \snd(\mathsf{a}),
          \snd(\mathsf{t'}))
          : \mathsf{B}
          }
          \quad
          \prftree{
          \mathsf{\Gamma} \vdash
          (\mathsf{t}, \mathsf{a}, \mathsf{t'})
          : \mathsf{A}
          }{
          \mathsf{\Gamma} \vdash
          (\mathsf{s}, \mathsf{b}, \mathsf{s'})
          : \mathsf{B}
          }{
          \mathsf{\Gamma} \vdash
          (\pair{\mathsf{t}}{\mathsf{s}},
          \pair{\mathsf{a}}{\mathsf{b}},
          \pair{\mathsf{t'}}{\mathsf{s'}})
          : \mathsf{A} \times \mathsf{B}
          }
          \\[5pt]
          \prftree{
          \mathsf{\Gamma} \vdash
          \mathsf{t} = \mathsf{s} : \mathsf{A}
          }{
          \mathsf{\Gamma}' \vdash
          \mathsf{t}' = \mathsf{s}' : \mathsf{A}
          }{
          \mathsf{\Gamma},
          \prt{\mathsf{\Gamma}} \vdash
          \mathsf{a} = \mathsf{b} : \prt{\mathsf{A}}
          }{
          \mathsf{\Gamma} \vdash
          (\mathsf{t},\mathsf{a},\mathsf{t}')
          : \mathsf{A}
          }{
          \mathsf{\Gamma} \vdash
          (\mathsf{s},\mathsf{b},\mathsf{s}')
          : \mathsf{A}
          }
        \end{array}
      \end{math}
    \end{minipage}
  }
  \caption{Derivation Rules}
  \label{fig:eqt}
\end{figure}

\section{A Lattice of \MMM{}s?}
\label{sec:towards-lattice-mmms}

We conclude our presentation with a few open questions about the relations holding between the different notions of program difference introduced in this paper.
When considering program equivalence, various non-equivalent notions have been 
introduced, such as observational equivalences, equivalences
derived from denotational semantics or equational
theories. Since observational equivalences are the
coarsest equivalences and equational theories are
the finest equivalences in many situations,
denotational semantics gives various mathematical
reasoning principles for observational
equivalences as well as equational theories.

Similarly, in the last sections we have introduced various notions of program differences, all
defined in terms of some form of differential logical relations. 
Therefore, it is reasonable to expect that a similar
comparison should be possible for \mmm{}s. In particular, this suggests the following two questions:
\begin{itemize}
\item Does the type indexed family $\dlog$
  give rise to
  the ``coarsest family of \mmm{}s''?
\item Does the type indexed family $\deq$ give
  rise to the ``finest family of \mmm{}s''?
\end{itemize}
We note that, although such differences are
defined over
$\cterm{\prt{\mathsf{A}}}$, which is not a
quantale, we can easily associate $\dlog$ and
$\deq$ with \mmm{}s valued on the quantale
$\mathcal{P}\cterm{\prt{\mathsf{A}}}$ of subsets
of $\cterm{\prt{\mathsf{A}}}$, letting
e.g.~$(\mathsf{t},a,\mathsf{t}') \in
\tilde{\delta}^{\mathrm{log}}_{\mathsf{A}} \iff
\text{for all }\mathsf{a} \in a,\,
(\mathsf{t},\mathsf{a},\mathsf{t}') \in
\dlog_{\mathsf{A}}$, and similarly for
${\deq_{\mathsf{A}}}$.


However, unfortunately, it is
not straightforward to tackle these questions
because of the two main obstacles.
First, while two \mmm{}s valued over the same quantale can be easily compared, it is not clear how to compare two \mmm{}s defined over \emph{different} quantales. 
 Second, while in the case of logical relations, the argument that logical equivalence is the coarsest one relies on a notion of \emph{observational equivalence}, it is not clear how to define a similar notion of \emph{observational quasi$^2$-metric} for $\LL$:   since differences between programs describe
  relationship between differences of inputs and
  differences of outputs, when we are to measure
  differences between programs, we should observe
  differences between outputs of programs with
  respect to different contexts. Therefore, we
  should define a notion of differences between
  contexts \emph{before} we define observational
  \mmm{} for $\LL$. How can we define differences
  between contexts?




\section{Related Work}

Differential logical relations for a simply typed
language were introduced in \cite{DGY19}, and
later extended to languages with monads \cite{DG22}, and related to incremental computing \cite{DLG21}.
Moreover, a unified framework
for operationally-based logical relations, subsuming differential logical relations, 
was introduced in \cite{lmcs:11041}.
The connections with metric spaces and partial
metric spaces have been explored already in
\cite{Geoffroy2020, PistoneLICS}, on the one hand
providing a series of negative results that
motivate the present work, and on the other hand
producing a class of metric and partial metric
models based on a different relational
construction.

The literature on the interpretation of linear or
graded lambda-calculi in the category of metric
spaces and non-expansive functions is ample
\cite{Reed2010, Gaboardi2013, Gaboardi2017,
  Gavazzo2018, Hoshino2023}. A related approach is
that of quantitative algebraic theories
\cite{Plotk}, which aims at capturing metrics over
algebras via an equational presentation. These
have been extended both to quantale-valued metrics
\cite{Dahlqvist2023} and to the simply typed
(i.e.~non graded) languages \cite{Honsell2022},
although in the last case the non-expansivity
condition makes the construction of interesting
algebras rather challenging.

The literature on partial metric spaces is vast,
as well. Introduced by Matthews \cite{matthews},
they have been largely explored for the
metrization of domain theory \cite{Bukatin1997,
  Schellekens2003, Smyth2006} and, more recently,
of $\lambda$-theories \cite{maestracci2025}. An
elegant categorical description of partial metrics
via the quantaloid of \emph{diagonals} is
introduced in \cite{Stubbe2018}. As this
construction is obviously related to the notion of
quasi-reflexivity here considered, it would be
interesting to look for analogous categorical
descriptions of the quasi$^2$-metrics here
introduced.

\section{Conclusion}
\label{sec:conclusion}

In this paper we have explored the connections
between the notions of program distance arising
from differential logical relations and those
defined via quasi-metrics and partial
quasi-metrics. As discussed in Section 9, our
results suggest natural and important questions
concerning the comparison of all the notions of
distance considered in this paper. At the same
time, our results provide a conceptual bridge that
could be used to exploit methods and results from
the vast area of research on quantale-valued
relations \cite{Hofmann2014, Stubbe2014} for the
study of program distances in higher-order
programming languages. For instance, natural
directions are the characterization of limits and,
more generally, of topological properties via
logical relations, as suggested by recent work
\cite{BCDG22}, although in a qualitative setting.

While we here focused on non-symmetric differential logical relations, 
understanding the metric structure of the symmetric case, as in \cite{DGY19}, would be interesting as well.
Notice that this would require to abandon quasi-reflexivity, cf.~Remark \ref{rem:symmetry}.

Finally, while in this paper we only considered simple
types, the notion of quasi$^2$-metric is robust
enough to account for other constructions like
e.g.~monadic types as in \cite{DG22}. It is thus
natural to explore the application of methods
arising from quasi-metric or partial quasi-metrics
for the study of languages with effects like
e.g.~probabilistic choice.


\bibliography{main}

\appendix

\end{document}